\documentclass[11pt, letterpaper]{article}

\usepackage[utf8]{inputenc}
\usepackage[T2A,T1]{fontenc} 
\usepackage[english]{babel}   
\usepackage{amsmath, amssymb, amsthm}
\usepackage{geometry}
\usepackage{hyperref}
\usepackage{enumitem}
\usepackage{algorithm}
\usepackage{algpseudocode}
\usepackage{tikz}
\usetikzlibrary{decorations.pathreplacing}

\hypersetup{
    colorlinks=true,
    linkcolor=blue,
    citecolor=red,
    pdftitle={Limit on the computational power of C-random strings},
    pdfauthor={Anonymous Author(s)} 
}

\newtheorem{theorem}{Theorem}[section]
\newtheorem{lemma}[theorem]{Lemma}
\newtheorem{definition}[theorem]{Definition}
\newtheorem{proposition}[theorem]{Proposition}

\newtheorem{corollary}[theorem]{Corollary}

\newcommand{\Str}{\{0,1\}^*}
\newcommand{\N}{\mathbb{N}}
\renewcommand{\C}{\mathrm{C}}

\newcommand{\K}{\mathrm{K}}

\newcommand{\cnd}{\mid}
\newcommand{\PP}{\mathrm{P}}
\newcommand{\m}{\mathbf{m}}
\newcommand{\ES}{\mathrm{EXPSPACE}}

\title{Limit on the computational power of $\C$-random strings}
\author{Alexey Milovanov}
\date{} 

\begin{document}

\maketitle

\begin{abstract}
We construct a universal decompressor $U$ for plain Kolmogorov complexity $\C_U$ such that the Halting Problem cannot be decided by any polynomial-time oracle machine with access to the set of random strings $R_{\C_U} = \{x : \C_U(x) \ge |x|\}$. This result resolves a problem posed by Eric Allender in~\cite{ea} regarding the computational power of Kolmogorov complexity-based oracles.
\end{abstract}
\tableofcontents
\newpage

\section{Introduction}
Informally the Kolmogorov complexity of a string $x$ is defined as the minimal length of a program that outputs $x$ on the empty input. This definition requires refinement; there are various types of complexity (plain, prefix, and others), but even once the type is fixed, there remains a dependency on the choice of the programming language or decompressor. We provide all the definitions needed for this paper in Section~\ref{sec:def}; we also refer the reader to~\cite{LiVit,suv} for formal definitions and the main properties of Kolmogorov complexity.

Kolmogorov complexity provides a notion of an individual random string: a string $x$
is considered random if its complexity is close to its length.
Let $R$ denote the set of all such random strings. A natural question is:
how powerful is the oracle $R$? For example, what languages belong to $\PP^R$?

This and similar questions were investigated in~\cite{abk, abkmr, sh}. In these papers, some \emph{lower bounds} for $\PP^R$, $\PP/\mathrm{poly}^R$, and other complexity classes were established. These results are robust in the following sense: they hold for all reasonable definitions of $R$; in particular, it does not matter what type of Kolmogorov complexity we consider.

How are these non-trivial lower bounds for classes such as $\mathrm{BPP}^R$
(which equals $\PP^R$~\cite{abkmr}) obtained?
The key observation is that access to $R$ allows a machine to distinguish
true randomness from pseudorandomness, effectively breaking any computable
pseudorandom generator. Since this argument works for essentially any
natural definition of the set of random strings $R$, it suggests that the
computational power of $\PP^R$ should be invariant across standard
definitions of $R$.

However, proving this invariance—and in particular, establishing matching \emph{upper bounds}—turns out to be exceptionally difficult. The situation for upper bounds is highly sensitive to the specific definition of complexity. The results in~\cite{afg, crelm, m} establish limits on the computational power of $R$ specifically for \emph{prefix} complexity. For example, the results in~\cite{afg, crelm} together with the lower bound obtained in~\cite{sh} show that
\[ \text{EXP}^{\text{NP}} \subseteq \bigcap_U \PP^{ R_{\text{K}_U}} \subseteq \ES. \]
Here, the intersection is taken over all universal prefix-free decompressors $U$. While there is a significant gap between the lower and upper bounds for prefix complexity, it is at least known that the intersection does not contain any undecidable language. 

For \emph{plain} complexity, even this basic limitation remained unknown. Prior to our work, it was an open question whether the Halting Problem could be decided by a polynomial-time machine with access to the plain complexity oracle. This paper provides a definitive step towards clarifying the situation for plain complexity.

Let $\C_U(x)$ denote the plain Kolmogorov complexity of a string $x$ with respect to a universal decompressor $U$. Instead of working directly with the set of random strings, we consider the stronger cumulative complexity oracle:
\[ F_{\C_U} = \{ \langle x, k \rangle \in \Str \times \N : \C_U(x) \le k \}. \]

Let $H$ denote the Halting Problem:
\[ H := \{ x \mid \text{program } x \text{ halts on the empty input}\}. \]

Our main result is the following theorem.

\begin{theorem}
\label{main}
There exists a universal decompressor $U$ such that $H \notin \PP^{F_{\C_U}}$.
\end{theorem}

\begin{corollary}
\label{cor:random_set}
Let $R_{\C_U} = \{x \in \Str : \C_U(x) \ge |x|\}$ be the set of random strings. Then $H \notin \PP^{R_{\C_U}}$.
\end{corollary}

\subsection*{Related works}
In~\cite{m25}, it was shown that for \emph{some} universal $U$ the Halting Problem belongs to $\PP^{R_{\C_U}}$. Together with our main result, this implies that the computational power of $\PP^{R_{\C_U}}$ significantly depends on the choice of the universal machine $U$.

In~\cite{kum}, Kummer showed that for each universal machine $U$, there is a time-bounded disjunctive truth-table reduction from $H$ to $R_U$. That is, there is a computable function that takes an input $x$ and produces a list of strings, with the property that $x \in H$ if and only if at least one of the strings is in $R_U$. Therefore, the polynomial restriction in Theorem~\ref{main} is significant.

\subsection*{Roadmap}
In Section~\ref{sec:overview}, we present a comprehensive technical overview of our framework, including the underlying combinatorial game and the formal statement of the Key Lemma. In Section~\ref{sec:def}, we recall the definition of plain Kolmogorov complexity and its basic properties.  Section~\ref{sec:constr}  proves the main theorem. Section~\ref{sec:key} contains the formal proof of the Key Lemma. Finally, omitted proofs of auxiliary statements are provided in the Appendix.

\section{Framework of the Proof}
\label{sec:overview}

Our goal is to construct a universal machine $U$
such that $H \notin \PP^{F_{\C_U}}$.
We first show that if the Halting Problem were solvable
in polynomial time relative to $F_{\C_U}$,
then a certain counting problem would also be solvable in polynomial time.
 The core of our paper is then to construct a specific $U$ that forces any polynomial-time oracle machine to fail at this task.

\subsection{From the Halting Problem to Counting}

Let $U$ be an arbitrary universal decompressor. We define the following task relative to $U$.

\begin{definition}[\texttt{COUNT-SIMPLE}]
    Given input $1^m$, output the integer $N_m$ representing the count of all strings with complexity at most $m$ with respect to $U$:
    \[ N_m = \left| \{ x \in \Str : \C_U(x) \le m \} \right|. \]
\end{definition}

Since any string $x$ with $\C_U(x) \le m$ is generated by a program of length at most $m$, we have the natural bound $0 \le N_m < 2^{m+1}$. 
The connection between the Halting Problem and this counting task is captured by the following lemma.

\begin{lemma}[Implication of $H \in \PP^{F_{\C_U}}$]
\label{lem:impl}
    If $H \in \PP^{F_{\C_U}}$, then there is a polynomial-time machine $M$ with oracle $F_{\C_U}$ that on input $1^m$ outputs $N_m$ (i.e., solves \texttt{COUNT-SIMPLE}).
\end{lemma}

\textit{Proof sketch.} Consider the predicate $Q(m, k)$: ``Do there exist at least $k$ strings $x$ with $\C_U(x) \le m$?'' If $H \in \PP^{F_{\C_U}}$, a polynomial-time machine can decide $Q(m, k)$. Using a standard binary search, the machine can pinpoint the exact value of $N_m$ in $O(m)$ queries. The formal proof is deferred to the Appendix. \qed

\subsection{The Architecture of the Universal Machine $U$}
\label{subsec:framework}
We construct our universal machine $U$ by combining a fixed optimal decompressor $V$ with a constructed partial computable function $G$. This architecture also depends on a constant $d \ge 1$. Both $G$ and $d$ will be formally determined in subsequent sections.

\begin{definition}[Universal Machine Structure]
\label{def:ums}
Given a fixed optimal universal decompressor $V$, a partial computable function $G$, and a constant $d \ge 1$, we define the machine $U$ on input $w \in \Str$ as follows:
\begin{enumerate}
    \item \textbf{Simulation Mode:} If $w = 0^d p$ (a string of $d$ zeros followed by $p$), then $U(w) = V(p)$.
    \item \textbf{Diagonalization Mode:} If $w = 1p$, then $U(w) = G(p)$.
\end{enumerate}
\end{definition}

This ensures that $U$ is a universal decompressor,
since $\C_U(x) \le \C_V(x)+d$ for all $x$.
This bound holds regardless of how we define $G$, provided $d \ge 1$, which  guarantees that the domains of the two modes are disjoint. Our core strategy relies on defining $G$ to create ``short'' programs (via the $1p$ prefix) for specific target strings, effectively bypassing the $d$-bit overhead to manipulate the oracle queries.

\paragraph{The Diagonalization Goal.}
Let $\{M_1, M_2, \dots\}$ be an enumeration of all polynomial-time oracle Turing machines, where each machine appears infinitely often.
Our goal is to construct $G$ so that for every machine $M_l$ there exists a specific input length $m$ such that:
\[
    M_l^{F_{\C_U}}(1^m) \neq N_m.
\]

\paragraph{Working Zones and the Tower Sequence.}
To manage the complex dependencies between the oracle queries made by $M_l$ and our dynamic construction of $U$, we partition the string lengths into distinct working zones. We define a rapidly growing sequence of lengths (the tower sequence). Let $a_0 = 1$ and $a_{l+1} = 2^{a_l}$.

For each step $l \ge 1$, we define the target input length $m$ for the diagonalization against machine $M_l$ as:
\[ m = n + d, \quad \text{where } n = a_{2l}. \]

\textit{Notation remark.}
We distinguish between the parameter $n$ (tied to $V$) and $m$ (tied to $U$).
Since $U$ simulates $V$ with a $d$-bit overhead, any string with $\C_V(x) \le n$
satisfies $\C_U(x) \le n+d = m$.
Thus, controlling complexities relative to $V$ determines the set of
$m$-simple strings for $U$.

We associate a specific ``working zone'' of lengths with each machine $M_l$, consisting of strings with lengths strictly between $a_{2l-1}$ and $a_{2l+1}$. Within this zone, we define the partial function $G$ to ensure that the count produced by $M_l$ is incorrect.

The logic of the construction is as follows:
\begin{enumerate}
    \item The machine $M_l$, on input $1^m$, may query the oracle $F_{\C_U}$ about various strings.
    \item We do not control the oracle's answers regarding small complexity values $k \le \log n = a_{2l-1}$. 
    \item We explicitly define $G$ for strings in the current zone (lengths from $\log n + 1$ to $2^n$) to force the discrepancy $M_l^{F_{\C_U}}(1^m) \neq N_m$.
\end{enumerate}

Crucially, our construction must succeed regardless of the oracle's fixed answers for the ``small'' complexities. The separation of these zones is illustrated in Figure~\ref{fig:zones}.

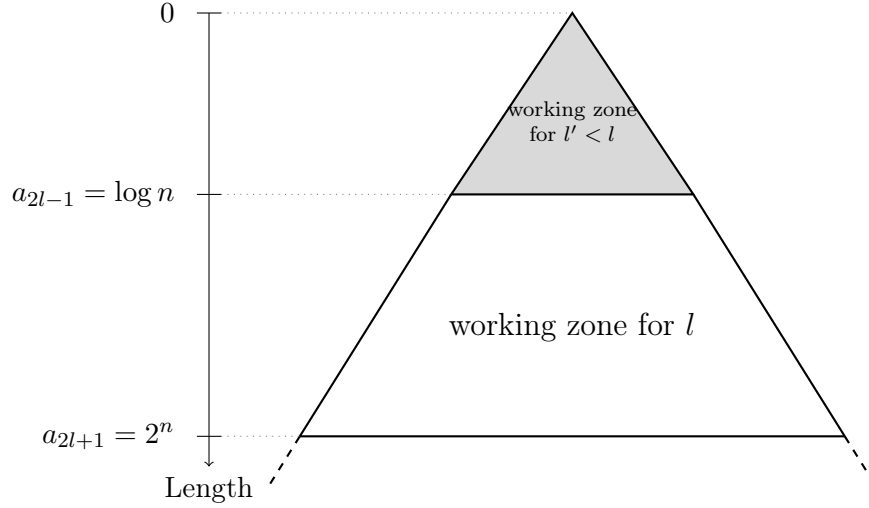
\begin{figure}[ht]
    \centering
    \begin{tikzpicture}[scale=0.8]
        \coordinate (Apex) at (0, 7);       
        \coordinate (MidL) at (-2, 4);      
        \coordinate (MidR) at (2, 4);       
        \coordinate (BotL) at (-4.5, 0);    
        \coordinate (BotR) at (4.5, 0);     

        \filldraw[fill=gray!30, draw=black, thick] 
            (Apex) -- (MidL) -- (MidR) -- cycle;

        \node[align=center, font=\scriptsize] at (0, 5.2) {working zone\\for $l' < l$};

        \draw[thick] (MidL) -- (BotL);      
        \draw[thick] (MidR) -- (BotR);      
        \draw[thick] (BotL) -- (BotR);      

        \node[align=center, font=\large] at (0, 1.8) {working zone for $l$};

        \draw[dashed, thick] (BotL) -- (-5, -0.8);
        \draw[dashed, thick] (BotR) -- (5, -0.8);

        \draw[->] (-6, 7) -- (-6, -0.5) node[below] {Length};
        
        \draw (-6.2, 7) -- (-5.8, 7);
        \draw (-6.2, 4) -- (-5.8, 4);
        \draw (-6.2, 0) -- (-5.8, 0);

        \node[left] at (-6.4, 7) {$0$};
        \node[left] at (-6.4, 4) {$a_{2l-1} = \log n$};
        \node[left] at (-6.4, 0) {$a_{2l+1} = 2^{n}$};

        \draw[dotted, gray] (-5.8, 7) -- (Apex);
        \draw[dotted, gray] (-5.8, 4) -- (MidL);
        \draw[dotted, gray] (-5.8, 0) -- (BotL);

    \end{tikzpicture}
    \caption{Structure of the construction. The gray area represents the history fixed by lower levels, while the white area is the active zone for level $l$.}
    \label{fig:zones}
\end{figure}

\subsection{The Diagonalization Game: A Heuristic Model}
\label{subsec:simple_game}

To explain the intuition behind the construction of $G$,
we view the diagonalization as a game between two players: \textbf{Alice} (representing the constructed function $G$) and \textbf{Bob} (representing the fixed optimal decompressor $V$), with the machine $M_l$ acting as a referee. 

\textit{Remark.}
This game is a simplified heuristic model.
The formal proof uses a refined version,
but the present formulation captures the main difficulty
and the intuition behind our construction.

\paragraph{The Setup and Goals.} 
Fix a polynomial-time oracle machine $M$ and an input length $m$. The game is initialized by setting the complexity of every string $x$ to its length: $\C(x) = |x| + O(1)$. 
The machine $M$ attempts to compute the target count $N_m = |\{x \mid \C(x) \le m\}|$ by querying the oracle ``Is $\C(x) \le k$?''.
\begin{itemize}
    \item \textbf{Bob's Goal:} Ensure the machine outputs the correct count, $M(1^m) = N_m$.
    \item \textbf{Alice's Goal:} Force the machine to output an incorrect count, $M(1^m) \neq N_m$.
\end{itemize}

\paragraph{The Rules and Budgets.} 
The game proceeds in rounds. In each round, a player can choose a string $x$ and a value $k < |x|$ to \textbf{reduce} its complexity to $k$ (effectively declaring $\C(x) = k$). Players are constrained by budgets derived from standard Kolmogorov complexity bounds:
\begin{itemize}
\item \textbf{Bob's Budget:}
for each $k$, he may assign complexity $\le k$
to at most $2^{k-d}$ strings.

\item \textbf{Alice's Budget:}
she cannot act for $k\le \log m$.
For $k>\log m$ she may assign complexity $\le k$
to at most $2^{k-1}$ strings.
\end{itemize}

The game ends when a player whose goal is currently failing has no valid moves left. If Alice has a winning strategy against \emph{any} allowable strategy of Bob, we can construct the partial function $G$ to satisfy Theorem~\ref{main}.

\paragraph{The Obstacle.}
It is unknown whether Alice has a winning strategy
in this basic game.
The difficulty lies in the asymmetric treatment
of small complexities ($k\le\log m$):
Bob fully controls these queries and may use them
to determine the behavior of $M$, potentially 
neutralizing Alice's moves on longer strings.

\paragraph{The Resolution.} 
The crucial insight is that Bob is not an arbitrary adversary; he represents the optimal decompressor $V$. Therefore, the canonical enumeration of his domain satisfies strict structural properties. In the next subsection, we introduce these structural properties, which transform this basic game into a strictly constrained setting where Alice \emph{does} have a provable winning strategy.

\subsection{Structural Properties of the Enumeration}
\label{subsec:structural_overview}

To formalize the constraints on Bob's moves, we analyze the canonical enumeration of the following set $\{ \langle x, k \rangle \mid \C_V(x) \le k \}$. 

\begin{definition}
Denote by $w_n$ the last string $w$  such that $\C_V(w) \le n$ that appears in this enumeration.
\end{definition}

For any $u > n$, let us analyze the ``tail'' of the enumeration for threshold $u$ relative to the moment $w_n$ appears.

\begin{definition}
Denote by $\operatorname{Rest}(n, u)$ the number of strings of complexity at most $u$ that appear in the enumeration \emph{after} $w_n$.
\end{definition}

We prove that this tail cannot be arbitrarily large. Specifically, if $u$ is computable from $n$, the remaining count is bounded by $2^{u - n + O(1)}$. 
To establish this, we first note the complexity property of the marker string $w_n$. Although the following propositions are likely known, we provide their formal proofs in the Appendix for the sake of completeness.

\begin{proposition}
\label{prop:wn_complexity}
The conditional complexity of $w_n$ satisfies:
\[ \C_V(w_n \cnd n) = n - O(1). \]
\end{proposition}
\begin{proposition}[Tail Upper Bound]
\label{prop:tail_bound_overview}
Let $n < u$ be integers such that $u$ is efficiently computable from $n$ (i.e., $\C_V(u \cnd n) = O(1)$). Then:
\[ \operatorname{Rest}(n, u) \le 2^{u - n + O(1)}. \]
\end{proposition}

Throughout the remainder of the paper, we denote the exponential threshold as $\mathcal{N} = 2^n$.
\begin{corollary}
\label{cor:exp_gap}
There exists a constant $c_1$ (dependent only on $V$) such that for every $n$:
\[ \operatorname{Rest}(n, \mathcal{N}) \le c_1 \cdot 2^{\mathcal{N} - n}. \]
\end{corollary}
\subsection{The Key Lemma: Scaling Up Complexity via Extensions}
\label{subsec:key}
For an arbitrary length $u$, the exact value of $\operatorname{Rest}(n, u)$ is equal to $2^{u - \K(w_n \cnd u) + O(1)}$, where $\K$ denotes prefix complexity. However, we cannot use such an expression involving uncomputable prefix complexity directly in our combinatorial construction.

Instead, for strings of arbitrary complexity, we perform a reduction to strings of complexity $\mathcal{N}$. More precisely, we assert that $\C_V(x) \le q$ if and only if there are ``many'' extension tuples of the form $(x, q, w, s)$ among the strings with complexity at most $\mathcal{N}$. The precise threshold for what constitutes ``many'' can be computed given $w_n$. This formal relationship is captured by the following Key Lemma.

\begin{lemma}[Key Lemma]
\label{lemma:unified_gap}
There exists a deterministic algorithm $\mathcal{A}$ with the following properties.

\textbf{Input:}
\begin{itemize}
    \item An integer $n$ such that $n = a_i$ for some $i$ in the tower sequence ($a_0=1, a_{i+1}=2^{a_i}$).
    \item A binary string $w$ (candidate for the last simple string $w_n$).
    \item A target string $x$ of length at most $n^{\log n}$.
    \item An integer $q$ such that $n = a_i \le q < n^{\log n}$ (the complexity threshold).
\end{itemize}
\textbf{Output:} An integer $B$.

\textbf{Properties:} 
The algorithm $\mathcal{A}$ always halts and produces an integer $B$. Furthermore, if $w$ is indeed the last string in the enumeration of complexity $\le n$ relative to $V$ (i.e., $w = w_n$), then the following implications hold.

Define for arbitrary integer $d$ the standard and relaxed extension sets relative to $V$:
\begin{align*}
    T_n(x, q, w) &= \{ s : \C_V(x, q, w, s) \le \mathcal{N} \} \\
    T^{\sqrt{d}}_n (x, q, w) &= \{ s : \C_V(x, q, w, s) \le \mathcal{N} + \sqrt{d} \}.
\end{align*}

Then, for all sufficiently large $d$:
\begin{enumerate}
    \item \textbf{Low Complexity implies High Count:}
    If the string is simple, it has many extensions.
    \[ \C_V(x) \le q \implies |T_n(x, q, w)| \ge B. \]
    
    \item \textbf{High Count (in relaxed set) implies Low Complexity:}
    If there are many extensions even in the relaxed sense, the string must be simple.
    \[ |T^{\sqrt{d}}_n (x, q, w)| \ge B \implies \C_V(x) \le q + \frac{d}{2}. \]
\end{enumerate}
\end{lemma}

\subsubsection*{Proof Idea of the Key Lemma}

The core idea of the proof is to establish a rigorous link between the plain complexity $\C_V(x)$ and the cardinality of the extension set $T_n (x, q, w)$. 
The strategy involves three main conceptual steps:

\begin{enumerate}
    \item \textbf{Translation to Prefix Complexity.} Both quantities can be expressed using conditional prefix complexity $\K$. First, the bound $\C_V(x) \le q$ is equivalent to $\K(x \cnd q) \le q$ (up to an additive constant). Second, the size of the extension set $|T_n(x, q, w)|$ is proportional to the discrete a priori probability $\m(x, q \cnd n)$ (which equals $2^{- \K(x, q \cnd n)}$ up to a constant factor). 
    
    \item \textbf{Bridging the Scales.} Our goal is to relate $\K(x \cnd q)$ and $\K(x, q \cnd n)$. First, observe that $\K(x \cnd q) = \K(x \cnd q, n) + O(1)$, because $q$ contains the information about $n$ (specifically, $n$ is determined by $q$ as the unique element $a_i$ of the tower sequence such that $a_i \le q < a_{i+1}$). 
    We then consider the following equality:
    \begin{equation}
    \label{eq_sim_overview}
     \K(x, q \cnd n) \approx \K(x \cnd q, n) + \K(q \cnd n).    
    \end{equation}
    However, we cannot use this equation directly for two reasons:
    \begin{itemize}
        \item \textbf{Precision:} This equation holds only up to a logarithmic error term.
        \item \textbf{Computability:} The prefix complexity function $\K(\cdot)$ is not computable. 
    \end{itemize}
    
    \item \textbf{The Solution.} We resolve these two fundamental issues as follows:
    \begin{itemize}
        \item \textbf{Oracle Power of $w$:} To handle non-computability, we utilize the candidate string $w$ (the last simple string of complexity $\le n$). Since the binary length of $q$ is logarithmic in $n$, the string $w$ acts as a powerful oracle, providing a universal time bound that allows us to compute $\K(q)$ and related terms exactly.
        \item \textbf{Fixed-Point Construction:} To handle the precision error, we replace the raw string $q$ with a structural proxy $\tilde{q} = \langle q, l \rangle$. We demonstrate, via a fixed-point argument, that there exists an integer $l$ such that an equation similar to~\eqref{eq_sim_overview} holds for the pair $\tilde{q}$ with strict $O(1)$ precision. Moreover, thanks to the oracle power of $w$, this exact $l$ can be found algorithmically given $q$ and $w$.
        \item \textbf{Volume Decomposition:} It remains to reduce the analysis of $q$ to our proxy $\tilde{q}$. The idea is to split all valid extension tuples $(x, q, s)$ into two distinct parts: those enumerated \emph{before} $w$, and those appearing \emph{after} $w$. The number of tuples appearing before $w$ can be computed algorithmically directly from $x$, $q$, and $w$. On the other hand, any tuple appearing \emph{after} $w$ inherently contains information about the history up to $w$ (and therefore about $q$), which in turn encodes $\tilde{q}$. This informational asymmetry cleanly completes the reduction.
    \end{itemize}
\end{enumerate}
\subsection{Game-Based Construction of $G$}
We present the construction of the partial function $G$. Recall our objective: to construct the universal decompressor $U$ such that for every polynomial-time oracle machine $M$, there exists a length $m$ where the output $M^{F_{\C_U}}(1^m)$ differs from the true count $N_m$.

The construction proceeds by partitioning the set of string lengths into distinct ``working zones.'' For a specific target length $m = n+d$ (derived from the tower sequence), the corresponding working zone consists of all lengths strictly between $\log n$ and $2^n$.

Our goal for a given level is to define $G$ for strings with lengths falling within this interval. We aim to define $G$ such that the machine $M = M_l$ associated with this level, upon receiving input $1^m$, produces an output distinct from $N_m$.

The function $G$ is constructed dynamically based on the behavior of the optimal decompressor $V$. We monitor the enumeration of all strings with $V$-complexity at most $n = m-d$ (note that their $U$-complexity is automatically at most $m$). Every time a new string $w$ with $\C_V(w) \le n$ appears, we start a new specific combinatorial game.

This game is played between two agents: \textbf{Alice} (representing the constructed function $G$) and \textbf{Bob} (representing the fixed decompressor $V$). Bob's goal is to ensure that the machine $M$, running with the oracle constructed so far, outputs the correct count $N_m$. Alice's goal is the opposite: to force $M$ to output an incorrect value.

The game resembles the ``Marking Game'' described in Subsection \ref{subsec:simple_game}, but with the following crucial modifications:
\begin{enumerate}
    \item \textbf{Restricted Simplification Levels:} Strings can only be simplified to complexity levels $q \ge m+1$.
    (This constraint reflects the assumption that the set of strings with complexity $\le m$ is effectively stabilized, and we focus on manipulating the oracle for queries larger than $m$.)
    
    \item \textbf{Resource-Based Budget:} There is no fixed budget for every complexity level. Instead, there is a single global budget derived from the ``Tail Upper Bound'' (Corollary~\ref{cor:exp_gap}):
    \[ \mathrm{Budget} = c_1 \cdot 2^{\mathcal{N} - n}. \]
    Furthermore, for every string $x$ and every target complexity $q$, there is a specific \textbf{price} $B(x, q)$ required to make $x$ $q$-simple. This price is determined by the threshold value $B$ from the \textbf{Key Lemma} (Lemma~\ref{lemma:unified_gap}), adjusting for the $d$-bit overhead between $U$ and $V$: 
    \[ B(x, q) = \mathcal{A}(n, w, x, q-d). \]   
\end{enumerate}
\subsubsection{The Local Game}
\label{subsec:comgame}
Here is the formal definition of the game.

\paragraph{Game Parameters.}
The game is played with the following fixed components:
\begin{itemize}
    \item \textbf{The Referee:} A polynomial-time oracle machine $M$ and a fixed input $1^m$.
    \item \textbf{The Target:} An integer $N_m$ (the count Bob defends).
    \item \textbf{Initial Complexity:} For every string $x$, there is an initial complexity $C_{\mathrm{init}}(x)$.
    \item \textbf{Prices:} A set of price thresholds $B(x, q) \in \mathbb{N}$ for pairs $(x, q)$.
    \item \textbf{Initial Counters:} Initial values $\mathrm{St}_{\mathrm{init}}(x, q) \in \mathbb{N}$ representing partial progress.
    \item \textbf{Player Budgets:} Each player is endowed with an initial budget $R$.
\end{itemize}

\paragraph{Dynamic State and Virtual Complexity.}
The state of the game at any step $t$ is defined by the current values of the counters $\mathrm{St}_t(x, q)$ and the remaining budgets of both players. Initially, $\mathrm{St}_0(x, q) = \mathrm{St}_{\mathrm{init}}(x, q)$.

The counters determine a ``virtual complexity'' function $C_t(x)$ according to the threshold rule:
\[
    C_t(x) = \min \left( C_{\mathrm{init}}(x), \quad \min \{ q \mid \mathrm{St}_t(x, q) \ge B(x, q) \} \right).
\]

In words, the complexity of string $x$ becomes $q$ as soon as the investment in the counter $\mathrm{St}(x, q)$ reaches the price $B(x, q)$. Note that $C_t(x)$ is non-increasing with respect to time $t$.

The oracle $F_t$ is derived directly from this complexity function:
\[ F_t = \{ \langle x, q \rangle \mid C_t(x) \le q \}. \]

\paragraph{The Game Protocol.}
The game proceeds in discrete turns. At each step $t$, we perform the following:

\begin{enumerate}
    \item \textbf{Evaluation:} We run the machine $M$ with the current oracle $F_t$ and observe the output $\mathrm{Out}_t = M^{F_t}(1^m)$.

    \item \textbf{Determine Active Player:} The player unsatisfied with the current result must move. Specifically, \textbf{Alice is active} if $\mathrm{Out}_t = N_m$ (the state favors Bob), and \textbf{Bob is active} if $\mathrm{Out}_t \neq N_m$ (the state favors Alice).
    
    \item \textbf{Action:} The active player chooses a set of non-negative increments $\Delta(x, q)$ to add to the counters:
    \[ \mathrm{St}_{t+1}(x, q) = \mathrm{St}_t(x, q) + \Delta(x, q). \]
    The total cost of the move is $\sum_{x, q} \Delta(x, q)$. This cost is deducted from the \textbf{active player's} remaining budget.
\end{enumerate}

\paragraph{Winning Condition.}
Since the budgets are finite and integer-valued, the game must eventually terminate. A player \textbf{loses} immediately if they are the active player but cannot make a move that changes the machine's output (either due to a lack of sufficient budget or because no such set of increments exists). Consequently, \textbf{Alice wins} if the game stabilizes at an output $\mathrm{Out} \neq N_m$, while \textbf{Bob wins} if it stabilizes at $\mathrm{Out} = N_m$.
\subsubsection{Informal Construction of $G$}
\label{subsubsec:inf}
We define the partial function $G$ implicitly through an auxiliary generative process called \textbf{Algorithm $\mathcal{F}$} (``The Fighter''). This algorithm runs in parallel with the fixed optimal decompressor $V$ and manages the diagonalization games for all levels simultaneously. Algorithm $\mathcal{F}$ outputs \textbf{directives} of the form ``Make string $z$ $q$-simple,'' which we implement by defining $G(p) = z$ for a fresh input $p$ of length $q-1$.
We interpret the output of $V$ as Bob's moves, and the definition of $G$ corresponds to Alice's moves.
\paragraph{Game Initialization and Budgets.}
$\mathcal{F}$ continuously monitors the enumeration of $V$. Whenever $V$ outputs a new string $w$ with complexity $\C_V(w) \le n$ (where $n= a_{2l}$ is the base threshold for level $l$), $w$ is treated as a new candidate for the true marker $w_n$. 
This event \textbf{resets} the Local Game at level $l$. At initialization, the prices $B(x, q)$ are computed via the Key Lemma. Crucially, both players are constrained by a global budget determined by the ``Tail Upper Bound'' (Corollary~\ref{cor:exp_gap}). 
Because levels are intertwined, we enforce a strict \textbf{Global Priority Rule}: whenever a new game starts at level $l$, or whenever $\mathcal{F}$ issues \emph{any} directive for level $l$, all active games at higher levels ($l' > l$) are immediately aborted.
\paragraph{The Winning Strategy Argument.}
If Alice has a winning strategy in this initialized game, she plays according to it, defining $G$ such that $M_l$ outputs a value different from $N_m$. (Note that the winner in this game can be determined algorithmically.) 

But what if Bob has a winning strategy? In this case, we employ a ``strategy stealing'' argument:
\begin{enumerate}
    \item Before the game begins, we use $G$ to create one additional $m$-simple string (a ``hidden'' string of length sufficiently large so that $M_l$ cannot query it) that is never queried by $M_l$.
    \item This action secretly increments the true target count: $N'_{m} = N_m + 1$.
    \item Now, Alice considers the game relative to the original target $N_m$. Since Bob had a strategy to force the output to be $N_m$, Alice can now \textbf{simulate Bob's strategy}. By playing ``as Bob,'' she forces $M_l$ to output $N_m$.
    \item However, the true target is now $N'_m = N_m + 1$. Since $M_l$ outputs $N_m$, the result is incorrect ($\mathrm{Out} \neq N'_m$). Thus, Alice achieves her goal.
\end{enumerate}
\paragraph{The Safety Check.}
To advance her strategy, Alice issues directives to lower the complexity of extension tuples, incrementing her counters $\mathrm{St}(x, q)$. When a counter reaches the required price $B(x, q)$, Alice is entitled to make the target string $x$ $q$-simple. 
$\mathcal{F}$ imposes a strict \textbf{Safety Check} before issuing the final directive for $x$:
\[ \text{Is } \C_V(x) \le q - d/2? \]
Only if this condition holds (or becomes true in the future) does $\mathcal{F}$ execute the simplification of $x$. Crucially, if the candidate $w$ is indeed the true last string, the second property of the \textbf{Key Lemma} (Lemma~\ref{lemma:unified_gap}) rigorously guarantees that this Safety Check will always eventually pass, allowing Alice's strategy to fully execute. Furthermore, this check guarantees that algorithm $\mathcal{F}$ will declare at most $2^{q - d/2 + 1} < 2^q$ strings as $q$-simple, which ensures that the partial function $G$ is well-defined.

\subsubsection*{Proof of the Main Theorem (sketch)}

Recall that by Lemma~\ref{lem:impl}, to show that $H \notin \PP^{F_{\C_U}}$, it suffices to show that for every polynomial-time oracle machine $M$, there exists an input length $m$ such that $M^{F_{\C_U}}(1^m) \neq N_m$.

Fix a machine $M = M_l$ from the enumeration. Consider the parameters for level $l$: $n = a_{2l}$ and $m = n + d$. Let $w$ be the \emph{last} string in the enumeration of $V$ with $\C_V(w) \le n$. Once $w$ appears, algorithm $\mathcal{F}$ identifies it as a candidate and initializes the final phase of the game at level $l$. At this moment, $N_m$ is fixed (possibly incremented by exactly one via ``Strategy Stealing''). From this point forward, $G$ executes \textbf{Alice's winning strategy}. We must verify that the interaction between $V$ (Bob) and $G$ (Alice) follows the Local Game rules, and that $G$'s moves are valid.

\paragraph{1. Bob (V) follows the Game Rules.}
\begin{itemize}
\item \textbf{Budget Constraint.} Bob's moves correspond to strings whose complexity drops below the threshold $\mathcal{N}$. By the \emph{Tail Upper Bound} (Corollary~\ref{cor:exp_gap}), the number of such strings appearing after $w$ is bounded by $c_1 \cdot 2^{\mathcal{N} - n}$, matching the finite budget allocated to the game.
\item \textbf{Consistency Constraint.} If $V$ outputs a description making $x$ simple (i.e., $\C_V(x)$ drops to $q-d$), Item~1 of the \textbf{Key Lemma} (Lemma~\ref{lemma:unified_gap}) ensures the number of simple extension tuples $(x, q-d, w, s)$ in $V$ is at least $B(x,q)$. Since $U$ simulates $V$, these tuples also exist in $U$. Thus, Bob's moves automatically satisfy the required threshold for counters $\mathrm{St}(x,q)$ and are consistent with the price mechanism.
\end{itemize}

\paragraph{2. Alice (G) makes Valid Moves.}
Since Alice follows a winning strategy, she forces the abstract game to an outcome $\mathrm{Out} \neq N_m$. However, for this to hold in the actual construction, every time Alice ``buys'' a query $\langle x,q\rangle$ by reaching $\mathrm{St}(x,q) \ge B(x,q)$, $G$ must successfully define a short description for $x$. This requires the \textbf{Safety Check} $\C_V(x) \le q - d/2$ to eventually hold. We claim that for the correct string $w$, this condition is always satisfied. Indeed:
\begin{enumerate}
\item Filling the counter $\mathrm{St}(x,q) \ge B(x,q)$ means Alice has created at least $B(x,q)$ extension tuples $(x, q-d, w, s)$ with $U$-complexity at most $\mathcal{N}+d$.
\item By definition of $U$, the number of strings with $\C_U \le \mathcal{N}+d$ is bounded by $2^{\mathcal{N}+O(1)}$. This includes strings with $V$-complexity at most $\mathcal{N}$ (at most $2^{\mathcal{N}+1}$) and strings generated by Alice's total budget across all games on this level ($2^{\mathcal{N}-n +O(1) }\cdot 2^n = 2^{\mathcal{N}  + O(1)}$).
\item This cardinality bound implies any string with $\C_U \le \mathcal{N}+d$ must satisfy $\C_V \le \mathcal{N} + O(\log d)$. Indeed, to describe any such string using the base machine $V$, one merely needs to specify the set itself (which takes $O(\log d)$ bits to encode the parameter $d$) along with the string's index within this set (which takes $\mathcal{N} + O(1)$ bits).
\item Consequently, Alice has produced at least $B(x,q)$ tuples whose $V$-complexity is at most $\mathcal{N}+O(\log d)$. Since $\sqrt{d} \gg O(\log d)$, applying Item~2 of the \textbf{Key Lemma} (Lemma~\ref{lemma:unified_gap}) with target complexity $q-d$ yields $\C_V(x) \le (q-d) + d/2 = q - d/2$.
\end{enumerate}
Therefore, whenever the strategy requires it, $G$ successfully enforces $\C_U(x) \le q$. Since Alice follows a winning strategy in the Local Game, the final output of the machine $M$ with oracle $F_{\C_U}$ differs from the true count $N_m$: $M^{F_{\C_U}}(1^m) \neq N_m$. This holds for every machine $M_l$ at its corresponding level, hence $H \notin \PP^{F_{\C_U}}$.

\section{Definitions and Basic Properties of Kolmogorov Complexity}
\label{sec:def}

We recall the definitions of plain conditional and unconditional Kolmogorov complexity.

Let $U$ be an algorithm whose inputs and outputs are binary strings. Such an algorithm is called a \emph{decompressor} and define the (unconditional) complexity $\C_U(x)$ of a binary string $x$ with respect to $U$ as follows:
\[ \C_U(x) := \min \{ |p|: U(p) = x \}. \]
We call any $p$ such that $U(p)=x$ a \emph{description} (or \emph{$U$-description}) of $x$. Thus, the complexity of $x$ is defined as the length of the shortest description of $x$.

A decompressor $U$ is called \emph{universal} if for every decompressor $U'$ there exists a constant $M$ such that for every string $x$,
\[ \C_U(x) \le \C_{U'}(x) + M. \]

Now let us recall the definition of \emph{conditional} Kolmogorov complexity.
Let $D(p, y)$ be a computable partial function (which we will also call a decompressor) of two string arguments. We may think of $D$ as an interpreter for some programming language, where the first argument $p$ is considered a program and the second argument $y$ is the input for this program. The conditional complexity is defined as:
\[ \C_D (x \cnd y) := \min \{|p|: D(p, y) = x\}. \]
Here $|p|$ denotes the length of the binary string $p$; thus, the right-hand side is the minimal length of a program that produces output $x$ given input $y$.

\begin{theorem}[Kolmogorov--Solomonoff]
There exists a universal conditional decompressor $U$ such that for every other conditional decompressor $D$, there exists a constant $M$ such that
\[ \C_U(x \cnd y) \le \C_{D} (x \cnd y) + M \]
for all strings $x$ and $y$.
\end{theorem}

The proofs for all statements in this section can be found in~\cite{suv}.

Note that unconditional complexity can be naturally derived from the conditional version by considering the condition to be the empty string. Let $U$ be a universal conditional decompressor; we define:
\[ \C_U(x) := \C_U(x \cnd \text{empty string}). \]
It is easy to verify that if $U$ is universal for conditional complexity, then the restricted decompressor is also universal for unconditional complexity.

We will use the following basic properties of Kolmogorov complexity, which hold for any fixed universal decompressor $U$:
\begin{itemize}
    \item \textbf{Bounded by length:} There exists a constant $M$ such that for all $x$:
    \[ \C_U(x) \le |x| + M. \]
    \item \textbf{Computable transformations:} For every total computable function $f$, there exists a constant $M$ such that for all $x$:
    \[ \C_U(f(x)) \le \C_U(x) + M. \]
    \item \textbf{Counting descriptions:} For every string $y$ and every integer $k$, the number of strings $x$ such that $\C_U(x \cnd y) \le k$ is at most $2^{k+1}$.
\end{itemize}

The Kolmogorov complexity $\C_U(x, y)$ of a pair of strings $x$ and $y$ is defined using a computable pairing function. Let $(x, y) \mapsto [x, y]$ be an injective computable function that maps a pair of strings to a single string. Then:
\[ \C_U(x, y) := \C_U([x, y]). \]
This definition depends on the choice of the pairing function, but only up to an additive $O(1)$ term. Similarly, one can define the complexity of triples of strings, integers, or any finite objects.

For every natural number $n$, its Kolmogorov complexity $\C_U(n)$ is at most $\log n + O(1)$ (corresponding to the length of its binary representation).

\subsection*{Prefix complexity}

Although the main result concerns plain Kolmogorov complexity,
we will use prefix complexity as a technical tool in the final section.
We briefly recall the necessary definitions and standard facts;
see~\cite{LiVit,suv} for details.

\paragraph{Prefix Complexity.}
Fix a universal prefix-free decompressor $D$.
The conditional prefix complexity is defined by
\[
\K(x \cnd y) := \min \{ |p| : D(p,y) = x \},
\]
where for every $y$ the domain of $D(\cdot,y)$ is prefix-free.

The unconditional version is $\K(x) := \K(x \cnd \varepsilon)$,
where $\varepsilon$ denotes the empty string.

\begin{proposition}[Basic properties {\cite[Ch.~3]{LiVit},\cite{suv}}]
For all strings $x,y,z$:
\begin{align*}
\K(x \cnd y,z) &\le \K(x \cnd y) + O(1), \\
\K(x,y \cnd z) &\ge \K(x \cnd z) + O(1), \\
\C_U(x \cnd y) &\le \K(x \cnd y) + O(1), \\
\K(x \cnd y) &\le \C_U(x \cnd y) + O(\log \C_U(x \cnd y)).
\end{align*}
\end{proposition}

\paragraph{Universal Semimeasure.}
A function $\mu(x \cnd y)$ is a \emph{conditional semimeasure}
if for every $y$,
\[
\sum_x \mu(x \cnd y) \le 1.
\]
It is \emph{lower semicomputable} if $\mu(x \cnd y)$
can be effectively approximated from below by rationals.

It is well known (see~\cite{LiVit,suv}) that there exists a maximal
lower semicomputable conditional semimeasure $\m(x \cnd y)$,
meaning that for every such $\mu$ there is a constant $c>0$ such that
\[
\mu(x \cnd y) \le c\,\m(x \cnd y)
\quad \text{for all } x,y.
\]
This $\m$ is called the \emph{universal semimeasure}
(or universal a priori probability).

\begin{theorem}[Coding Theorem {\cite[Th.~4.3.3]{LiVit}}, {\cite{suv}}]
\label{theorem:coding_theorem}
For all strings $x,y$,
\[
\K(x \cnd y) = -\log \m(x \cnd y) + O(1).
\]
\end{theorem}

\begin{theorem}[Chain Rule {\cite[Th.~3.9.1]{LiVit}}]
\label{theorem:chain_rule}
For all strings $x,y,z$,
\[
\K(x,y \cnd z)
=
\K(x \cnd z)
+
\K(y \cnd x, \K(x \cnd z), z)
+
O(1).
\]
\end{theorem}

\begin{corollary}
\label{cor:ch}
For all strings $x,y,z$,
\[
\K(y \cnd z)
\le
\K(x \cnd z)
+
\K(y \cnd x)
+
O(1).
\]
\end{corollary}

\begin{proposition}[Complexity Self-Sufficiency]
\label{prop:self_sufficiency}
For all strings $x,y$,
\[
\K(x, \K(x \cnd y) \cnd y)
=
\K(x \cnd y)
+
O(1).
\]
\end{proposition}

\begin{proof}
Apply Theorem~\ref{theorem:chain_rule} with $z=y$
and let $k=\K(x \cnd y)$. Then
\[
\K(x,k \cnd y)
=
\K(x \cnd y)
+
\K(k \cnd x, \K(x \cnd y), y)
+
O(1).
\]
Since $k$ is explicitly given in the condition,
the second term is $O(1)$.
\end{proof}

\begin{proposition}[Plain vs.\ Prefix {\cite[Th.~72]{suv}}]
\label{prop:plain_prefix}
For every string $x$,
\[
\C_U(x) = \K(x \cnd \C_U(x)) + O(1).
\]
\end{proposition}

\begin{corollary}[Equivalence of Bounds]
\label{cor:equivalence_bounds}
For every string $x$ and integer $q$:
\begin{enumerate}
    \item If $\C_U(x) \le q$, then $\K(x \cnd q) \le q + O(1)$.
    \item If $\K(x \cnd q) \le q$, then $\C_U(x) \le q + O(1)$.
\end{enumerate}
\end{corollary}

\begin{proof}
\textbf{1.}
Assume $\C_U(x) \le q$ and let $k = \C_U(x)$.
By Corollary~\ref{cor:ch},
\[
\K(x \cnd q)
\le
\K(x \cnd k)
+
\K(k \cnd q)
+
O(1).
\]
By Proposition~\ref{prop:plain_prefix},
$\K(x \cnd k) = k + O(1)$.
Let $\delta = q - k \ge 0$.
Since $k$ is easily computable from $q$ and $\delta$, we have $\K(k \cnd q) \le \K(\delta \cnd q) + O(1)$.
Since $\K(\delta \cnd q) \le 2\log \delta + O(1)$,
we obtain
\[
\K(x \cnd q)
\le
k + 2\log(q-k) + O(1)
\le
q + O(1).
\]

\textbf{2.}
Assume $\K(x \cnd q) \le q$.
Let $p$ be a shortest prefix-free program producing $x$ from $q$,
and let $|p|=l \le q$.
Put $\delta = q-l \ge 0$.

To describe $x$ for the plain machine,
it suffices to provide a self-delimiting (prefix-free) encoding of the integer $\delta$ followed by the raw string $p$. The length of this combined description is bounded by $2\log\delta + l + O(1)$.
Then
\[
\C_U(x)
\le
l + 2\log\delta + O(1)
=
q - \delta + 2\log\delta + O(1)
\le
q + O(1),
\]
since $2\log\delta - \delta$ is bounded from above.
\end{proof}

\section{Game-Based Construction of $G$}
\label{sec:constr}
\subsection{Definition of $G$}
\label{subsec:def_G}

As informally outlined in Subsubsection~\ref{subsubsec:inf}, we define the partial function $G$ implicitly through an auxiliary generative process called \textbf{Algorithm $\mathcal{F}$} (``The Fighter''). This algorithm runs in parallel with the fixed optimal decompressor $V$ and manages the diagonalization games for all levels simultaneously based on a priority system.

The algorithm $\mathcal{F}$ outputs \textbf{directives}---instructions to assign specific complexities to strings. All directives have the following structure:
\[ \text{``Make string $z$ $q$-simple.''} \]
The function $G$ is determined by these directives. To implement a directive that requires a string $z$ to have complexity $q$, we consume the next available input string $p$ of length $q-1$ and define $G(p) = z$.
(In the next subsection, we will prove that for any $q$, $\mathcal{F}$ issues at most $2^{q-1}$ simplification directives, ensuring $U$ remains a valid decompressor.)

Let $\{M_1, M_2, \dots\}$ be an enumeration of all polynomial-time oracle Turing machines such that every machine appears in the sequence infinitely often.

\paragraph{Levels and Status.}
We partition the problem into \textbf{Levels}. Level $l$ is responsible for handling strings with complexities $q$ in the interval:
\[ 
a_{2l-1} < q \le a_{2l+1}. 
\]
For each level $l$, we denote the ``base threshold'' as $n = a_{2l}$.
At any point in time, each level $l$ is in one of two states:
\begin{itemize}
    \item \textbf{Waiting (Standby):} The level is inactive, waiting for a stable candidate $w$ from $V$.
    \item \textbf{Active (Game):} The level has a valid candidate $w$ and is actively playing the local game against machine $M_l$.
\end{itemize}

\paragraph{The Event Loop.}
Algorithm $\mathcal{F}$ monitors the output of $V$. Suppose $V$ outputs a new description of length $q$ for a string $x$ (establishing a new current upper bound $\C_V(x) \le q$).
We determine the level $l$ such that $a_{2l-1} < q \le a_{2l+1}$. Let $n = a_{2l}$. We distinguish two cases based on the value of $q$.

\subsubsection*{Case 1: New Candidate ($q \le n$)}
If $q \le n$, we treat this string (let's call it $w$) as a new candidate for the \emph{last} string of complexity $\le n$ in the enumeration of $V$.

The appearance of a new candidate invalidates any previous game at this level. We switch the status of Level $l$ to \textbf{Active} and initialize a new \textbf{Local Game} (as described in Section~\ref{subsec:comgame}) with the following parameters:

\begin{enumerate}
    \item \textbf{The Referee:} The polynomial-time oracle machine $M_l$ (from our enumeration) running on fixed input $1^m$, where $m = n+d$.

    \item \textbf{The Target:} The current value of $N_m$ (the current number of strings whose $U$-complexity does not exceed $m$).

    \item \textbf{Initial Complexity:} For every string $x$, we set $C_{\mathrm{init}}(x)$ to be its current complexity $\C_U(x)$ (determined by $V$ and previous directives).

    \item \textbf{Prices:} We invoke algorithm $\mathcal{A}$ (Lemma~\ref{lemma:unified_gap}) to compute thresholds:
    \[ B(x, q) = \mathcal{A}(n, w, x, q-d). \]

    \item \textbf{Initial Counters:} We initialize the counters $\mathrm{St}_{\mathrm{init}}(x, q)$ by calculating the number of extension tuples already present in the enumeration of $U$ up to the current computational step. Specifically:
    \[ 
    \mathrm{St}_{\mathrm{init}}(x, q) = \left| \{ s : \C_U(x, q-d, w, s) \le \mathcal{N}_d \} \right|, 
    \]
    where $\mathcal{N}_d = 2^n + d$.

    \item \textbf{Player Budgets:} The global budget for each player is set to:
    \[ R = c_1 \cdot 2^{\mathcal{N} - n}, \]
    where $c_1$ is the constant from Corollary~\ref{cor:exp_gap}.
\end{enumerate}

\paragraph{Algorithmic Computability of the Winner.}
We claim that the winner of this initialized game can be deterministically computed by algorithm $\mathcal{F}$. 
First, observe that the game must conclude in a finite number of turns. Both players have strictly bounded finite budgets, and the rules dictate that the player currently unsatisfied with the referee's output must make a move. Since every valid move costs at least one unit of budget, the game must eventually terminate.

However, there is a subtle point: at first glance, the game tree might appear to have an infinite branching factor because a player can choose an arbitrary string $x \in \Str$ to simplify, implying an infinite action space. We resolve this by reducing the game to an equivalent finite version. 

Recall that the referee machine $M_l(1^m)$ runs in polynomial time, meaning it can only make oracle queries of length bounded by $\mathrm{poly}(m)$. In particular, $M_l$ can never query the oracle about any string $x$ of length $|x| > 2^m$. From the perspective of the machine $M_l$, all such ``long'' strings are completely indistinguishable. 

Therefore, to algorithmically determine the winner, algorithm $\mathcal{F}$ analyzes a modified state space. For ``short'' strings ($|x| \le 2^m$), the exact identity and complexity of each string are tracked (this is a finite set). For ``long'' strings ($|x| > 2^m$), the specific identities are ignored; instead, the state merely maintains a counter tracking \emph{how many} long strings have been made $q$-simple for each valid $q$. 

In this reduced game, the action of making a specific long string simple is replaced by the action of incrementing the corresponding generic counter. This makes the set of available moves at any step strictly finite. Since the overall game is now a finite perfect-information game, $\mathcal{F}$ can algorithmically evaluate the entire game tree (e.g., via the minimax algorithm) to determine the winning strategy from the initial configuration.

Based on the outcome of this algorithmic evaluation, we determine the \textbf{Strategy}:
\begin{itemize}
    \item \textbf{If Bob wins} (i.e., Bob has a winning strategy to force the game to stabilize at output $N_m$ despite Alice's best moves): We apply \textbf{Strategy Stealing}. We select a ``dummy'' string $z$ sufficiently long to be unqueried by $M_l$ (e.g., $|z| > 2^m$) and issue a directive to make $z$ $m$-simple. 
    We claim that after this single modification, Alice now possesses a constructive winning strategy. In fact, she has a winning strategy for an even stricter goal: achieving exactly $M_l(1^m) + 1 = N'_m$, where $N'_m = N_m + 1$ is the new true target count.
    The reasoning is as follows. The directive to make $z$ $m$-simple increases the true number of $m$-simple strings by exactly one, so the new target is $N'_m = N_m + 1$. However, because $z$ is too long to be queried by $M_l$, the machine's execution path and output remain completely oblivious to this specific simplification. Therefore, Alice can simply abandon her original plan and adopt Bob's winning strategy from the original game. By playing ``as Bob,'' Alice forces the machine to stabilize at the original target output: $M_l(1^m) = N_m$. Since the true target is now $N'_m = N_m + 1$, the machine's output is incorrect, and Alice successfully achieves her diagonalization goal.

    A minor technicality arises: what if Bob's original winning strategy required making this exact string $z$ simple at some later point? This is resolved by recalling that the strategy is derived from the \emph{reduced} game. In the reduced game, moves involving long strings do not specify the exact identity of the string; they merely dictate the action of making \emph{some} long string simple (incrementing a generic counter). When translating this strategy back to the actual construction, Alice can fulfill these directives by selecting \emph{any} fresh, unqueried long string. Since the game has a strictly bounded finite budget and the supply of long strings ($|x| > 2^m$) is infinite, Alice can always find a new dummy string to execute Bob's moves without ever colliding with her initial target-flipping string $z$.
    \item \textbf{If Alice wins} (or after Strategy Stealing): We proceed to execute Alice's strategy (described below).
\end{itemize}

\textbf{Priority Enforcement:} Since a new game has started at level $l$, we immediately \textbf{abort} all games at higher levels ($l' > l$). These levels switch to \textbf{Waiting} status.

\subsubsection*{Case 2: Non-Candidate Output ($q > n$)}
If $V$ outputs a new description of length $q > n$ (within the level's range), the reaction depends on the current status of Level $l$:
\begin{itemize}
    \item \textbf{Status: Waiting.} We do nothing.
    \item \textbf{Status: Active.} We continue executing Alice's winning strategy based on the updated state.
\end{itemize}

\paragraph{Executing Alice's Strategy.}
When the level is Active, $\mathcal{F}$ generates directives to advance Alice's position.
\begin{itemize}
    \item \textbf{Increments (Fueling):} If the strategy requires increasing a counter $\mathrm{St}(x, q)$, we issue a directive to \textbf{make the string encoding the tuple $(x, q-d, w, s)$ $\mathcal{N}_d$-simple}, where $s$ is a new unique index.
    
    \item \textbf{Payoff (Simplification):} If a counter $\mathrm{St}(x, q)$ reaches the price $B(x, q)$, we add $x$ to a \emph{preliminary list} of simple strings. We then monitor the following \textbf{Safety Check} condition:
    \[ \C_V(x) \le q - d/2. \]
    If this condition holds currently or becomes true at any point in the future, $\mathcal{F}$ issues a directive to \textbf{make $x$ $q$-simple}.
    (If $w$ is truly the last string, Lemma~\ref{lemma:realization} below guarantees this check will pass.)
\end{itemize}

\textbf{Global Priority Rule:} Whenever $\mathcal{F}$ issues \emph{any} directive for Level $l$, the global complexity landscape changes. Consequently, we immediately \textbf{abort} all games at higher levels ($l' > l$), switching them to \textbf{Waiting} status.

\subsection{Correctness of the Construction: Resource Accounting}

To prove that the function $G$ is well-defined, we must demonstrate that the algorithm $\mathcal{F}$ never runs out of input strings. Since $G$ maps inputs $p$ to outputs $z$ to enforce complexity constraints $\C_U(z) \le |p|$, we must ensure that for any target complexity length $q$, the number of inputs of length $q-1$ consumed by the construction is strictly bounded by the total number of available strings ($2^{q-1}$).

We analyze the three types of operations consuming input strings. We assume the constant $d$ is chosen sufficiently large to satisfy the inequalities derived below.

\paragraph{1. Dummy Strings (Target complexity $m$).}
These strings are used in the ``Strategy Stealing'' step (Target Flipping). We define a new dummy string only when a local game phase restarts (i.e., a new candidate $w$ appears).
\begin{itemize}
    \item \textbf{Consumption:} The number of phases corresponds to the number of strings $w$ such that $\C_V(w) \le n$ (where $n = m-d$). This count is bounded by $2^{n+1}$.
    \item \textbf{Available Space:} The inputs used here have length $m-1 = n+d-1$. Thus, there are $2^{n+d-1}$ available strings.
    \item \textbf{Check:} We require $2^{n+1} < 2^{n+d-1}$. This inequality holds for any $d \ge 3$.
\end{itemize}

\paragraph{2. Tuple Strings (Target complexity $\mathcal{N}_d$).}
These strings are used in the ``Fuel'' step to increment counters.
\begin{itemize}
    \item \textbf{Consumption:} The total number of increments is the product of the number of games played and the budget allocated to each game.
    \begin{itemize}
        \item Max number of games (candidates $w$): $2^{n+1}$.
        \item Budget per game (from Corollary~\ref{cor:exp_gap}): $c_1 \cdot 2^{\mathcal{N} - n}$.
    \end{itemize}
    Multiplying these, we get the total consumption bound:
    \[ \text{Total Fuel} \le 2^{n+1} \cdot c_1 \cdot 2^{\mathcal{N} - n} = 2 c_1 \cdot 2^{\mathcal{N}}. \]
    \item \textbf{Available Space:} The inputs used here have length $\mathcal{N}_d - 1 = \mathcal{N} + d - 1$. Thus, the capacity is $2^{\mathcal{N} + d - 1}$.
    \item \textbf{Check:} We need to ensure:
    \[ 2 c_1 \cdot 2^{\mathcal{N}} < 2^{\mathcal{N} + d - 1}. \] This inequality clearly holds for sufficiently large $d$.
\end{itemize}

\paragraph{3. Realization Strings (Target complexity $q$).}
These strings are used in the ``Payoff'' step to declare a string $x$ to be $q$-simple.
\begin{itemize}
    \item \textbf{Consumption:} The algorithm $\mathcal{F}$ issues a directive for $x$ \emph{only if} the safety check passes: $\C_V(x) \le q - d/2$. The number of such strings $x$ is bounded by $2^{q - d/2 + 1}$.
    \item \textbf{Available Space:} There are $2^{q-1}$ input strings of length $q-1$.
    \item \textbf{Check:} We require $2^{q - d/2 + 1} < 2^{q-1}$. This simplifies to $d > 4$.
\end{itemize}

\paragraph{Conclusion.}
By choosing a sufficiently large constant $d$ (dependent only on the base machine $V$), we guarantee that for every relevant length, the construction utilizes only a negligible fraction of the available input strings. Thus, $G$ is well-defined and can fully realize Algorithm $\mathcal{F}$.

\subsection{Proof of the Main Theorem}

We are now ready to prove Theorem~\ref{main}.

The decompressor $U$ is defined in Definition~\ref{def:ums}, and the function $G$ is defined in Section~\ref{subsec:def_G}. The constant $d$ is chosen sufficiently large to satisfy the resource accounting conditions from the previous subsection, as well as the conditions derived below.

Recall that by Lemma~\ref{lem:impl}, to show that $H \notin \PP^{F_{\C_U}}$, it suffices to show that for every polynomial-time oracle machine $M$, there exists an input length $m$ such that $M^{F_{\C_U}}(1^m) \neq N_m$, where $N_m = |\{x : \C_U(x) \le m\}|$.

Fix a machine $M = M_l$ from the enumeration. Consider the parameters for level $l$: $n = a_{2l}$ and $m = n + d$. Let $w$ be the \emph{last} string in the enumeration of $V$ with $\C_V(w) \le n$. 
Once $w$ appears, the algorithm $\mathcal{F}$ identifies it as a candidate and initializes the final phase of the game at level $l$. At this moment, the target count $N_m$ is fixed (possibly incremented by exactly one via the ``Strategy Stealing'' step). From this point forward, the construction of $G$ executes \textbf{Alice's winning strategy}.

\paragraph{Stability against Lower Stages.}
Note that the candidate string $w$ is produced by $V$ and satisfies $\C_V(w) \ge n - O(1)$ (by Proposition~\ref{prop:wn_complexity}).
Any string $x$ affected by games at lower levels $l' < l$ has complexity $O(\log n)$, which is drastically smaller than $n$. Thus, the candidate $w$ is distinct from any artifact produced by lower stages.

We must verify that the interaction between $V$ (Bob) and $G$ (Alice) follows the Local Game rules, and that $G$'s moves are valid.
\paragraph{1. Bob (V) follows the Game Rules.}
We interpret the subsequent output of $V$ as moves by Bob.
\begin{itemize}
    \item \textbf{Budget Constraint.} Bob's moves correspond to strings whose complexity drops below the threshold $\mathcal{N} = 2^n$. By the \emph{Tail Upper Bound} (Corollary~\ref{cor:exp_gap}), the number of such strings appearing after $w$ is bounded by $c_1 \cdot 2^{\mathcal{N} - n}$, matching the finite budget allocated to the game.
    \item \textbf{Consistency Constraint.} If $V$ outputs a description making $x$ simple (i.e., $\C_V(x)$ drops to $q-d$), Item~1 of the \textbf{Key Lemma} (Lemma~\ref{lemma:unified_gap}) ensures the number of simple extension tuples $(x, q-d, w, s)$ in $V$ is at least $B(x,q)$. Since $U$ simulates $V$, these tuples also exist in $U$. Thus, Bob's moves automatically satisfy the required threshold for counters $\mathrm{St}(x,q)$ and are consistent with the price mechanism.
\end{itemize}
\paragraph{2. Alice (G) makes Valid Moves.}
Since Alice plays a winning strategy, she forces the abstract game to an outcome $\mathrm{Out} \neq N_m$. However, for this to hold in the actual construction, every time Alice ``buys'' a query $\langle x, q \rangle$ by reaching the counter threshold $\mathrm{St}(x, q) \ge B(x, q)$, the function $G$ must successfully define a short description for $x$. This requires the following \textbf{Safety Check} to eventually hold:
\[ \C_V(x) \le q - d/2. \]
We prove that for the correct string $w$, this check always passes.

\begin{lemma}[Realization Lemma]
\label{lemma:realization}
Let $w$ be the true last string of $V$-complexity $\le n$. Let $x$ be a string and $q$ be a threshold such that Alice has filled the reservoir with tuples in $U$:
\[ |\{ s : \C_U(x, q-d, w, s) \le \mathcal{N}_d \}| \ge B(x, q), \]
where $\mathcal{N}_d = 2^n + d$ and $B(x, q) = \mathcal{A}(n, w, x, q-d)$.
Then, for all sufficiently large $d$:
\[ \C_V(x) \le q - d/2. \]
\end{lemma}

\begin{proof}
We establish the bound on $\C_V(x)$ by analyzing the set of all strings with bounded $U$-complexity through the following four steps:

\begin{enumerate}
    \item \textbf{Creation of Tuples:} Filling the counter $\mathrm{St}(x,q) \ge B(x,q)$ means that Alice has successfully defined at least $B(x,q)$ extension tuples of the form $(x, q-d, w, s)$ such that their $U$-complexity is at most $\mathcal{N}_d$.
    
    \item \textbf{Bounding the Size of the Bounded-$U$-Complexity Set:} Consider the complete set of strings $S = \{ y : \C_U(y) \le \mathcal{N}_d \}$. By the definition of the universal decompressor $U$, any string in $S$ originates from one of two sources:
    \begin{itemize}
        \item \emph{Simulation of $V$:} Since $U$ simulates $V$ with a length overhead of $d$, a $U$-description of length at most $\mathcal{N}_d = \mathcal{N} + d$ corresponding to $V$'s simulation comes from a $V$-description of length at most $\mathcal{N}$. The number of such strings is strictly bounded by $2^{\mathcal{N}+1}$.
        \item \emph{Directives of $G$:} Strings generated by Alice across all games on this level. As shown in the Resource Accounting section, the total budget used is bounded by the maximum number of candidates ($2^{n+1}$) multiplied by the global budget per game ($O(2^{\mathcal{N}-n})$), yielding at most $O(2^{\mathcal{N}})$ strings.
    \end{itemize}
    Summing these contributions, the total size of $S$ is bounded by $2^{\mathcal{N} + O(1)}$.

    \item \textbf{Bounding the $V$-complexity:} Since $V$ is an optimal universal decompressor, it can algorithmically simulate the enumeration of the set $S$. To describe any specific string $y \in S$ using $V$, one merely needs to specify the set $S$ itself, followed by the index of $y$ within $S$. 
    Specifying the set $S$ only requires encoding the parameter $d$ (which takes $O(\log d)$ bits), since the level parameters $l$ and $n$ are fixed. Describing the index within $S$ requires $\log |S| \le \mathcal{N} + O(1)$ bits. Therefore, the total $V$-complexity for any $y \in S$ satisfies:
    \[ \C_V(y) \le \mathcal{N} + O(\log d). \]

    \item \textbf{Applying the Key Lemma:} From Step 1, Alice has produced at least $B(x,q)$ extension tuples that belong to $S$. By Step 3, each of these tuples has a $V$-complexity of at most $\mathcal{N}+O(\log d)$. Since $\sqrt{d} \gg O(\log d)$ for sufficiently large constants $d$, these tuples fall comfortably within the relaxed extension set:
    \[ T^{\sqrt{d}}_n(x, q-d, w) = \{ s : \C_V(x, q-d, w, s) \le \mathcal{N} + \sqrt{d} \}. \]
    Applying Item~2 of the \textbf{Key Lemma} (Lemma~\ref{lemma:unified_gap}) with target complexity $q-d$, we conclude:
    \[ \C_V(x) \le (q-d) + d/2 = q - d/2. \]
\end{enumerate}
\end{proof}
\paragraph{Conclusion.}
The Realization Lemma ensures that every move dictated by Alice's strategy passes the safety check in the definition of $G$. Consequently, $G$ successfully forces $\C_U(x) \le q$ whenever the strategy requires it.
Since Alice follows a winning strategy in the Local Game, the final output of the machine $M$ with oracle $F_{\C_U}$ will differ from the target count $N_m$:
\[ M^{F_{\C_U}}(1^m) \neq N_m. \]
This holds for every machine $M_l$ at its corresponding level. Thus, $H \notin \PP^{F_{\C_U}}$.
\qed

\section{Proof of the Key Lemma}
\label{sec:key}

In this section, we present the complete formal proof of the Key Lemma (Lemma~\ref{lemma:unified_gap}), which was originally formulated and sketched in Section~\ref{subsec:key}.

Before diving into the technical details, we offer a brief note of encouragement to the reader. While this section is quite lengthy, its size is a deliberate choice to prioritize clarity over brevity. We have modularized the main argument into several self-contained sub-lemmas. Most of these intermediate results rely on standard, well-established techniques from algorithmic information theory. The primary exception is the $\tilde{q}$-Structure Lemma (Lemma~\ref{lemma:q_tilde}), which introduces a novel structural decomposition. By meticulously spelling out every transition, bound, and algebraic manipulation, we aim to make the verification process straightforward and the underlying mechanics fully transparent.

Before detailing the proof strategy, we introduce the necessary notation.
Throughout this section, we use the notation $\stackrel{*}{=}$, $\stackrel{*}{\le}$, and $\stackrel{*}{\ge}$ to denote equalities and inequalities that hold up to a multiplicative constant factor depending only on the universal machine $V$ (i.e., $A \stackrel{*}{=} B$ means $c_1 A \le B \le c_2 A$ for some constants $c_1, c_2 > 0$).

We define the complexity thresholds for the two scales of interest:
\[ \mathcal{N} = 2^n \quad \text{and} \quad \mathcal{Q} = 2^q. \]

We define the sets of valid extensions relative to these thresholds. For any string (or tuple) $y$, we denote:
\begin{align*}
    T_n(y) &= \{ s \mid \C_V(y, s) \le \mathcal{N} \}, \\
    T_q(y) &= \{ s \mid \C_V(y, s) \le \mathcal{Q} \}.
\end{align*}
\subsection{Measure vs. Extension Density}
\label{subsec:measure}
We establish a correspondence between the conditional a priori probability $\m(y \cnd n)$ and the number of simple extensions. Denote the relaxed thresholds for arbitrary $\delta \ge 0$:
\[ T^\delta_n(y) = \{ s : \C_V(y, s) \le \mathcal{N} + \delta \}. \]
\[ T^\delta_q(y) = \{ s : \C_V(y, s) \le \mathcal{Q} + \delta \}. \]

\begin{lemma}[Measure Density Lemma]
\label{lemma:measure_density}
Let $y$ be a string with length $|y| < \mathcal{N}/2 - O(1)$.

\begin{enumerate}
    \item \textbf{Standard Threshold:}
    \begin{align*}
        \m(y \cnd n) &\stackrel{*}{\le} |T_n(y)| \cdot 2^{-\mathcal{N}}, \\
        \m(y \cnd q) &\stackrel{*}{\le} |T_q(y)| \cdot 2^{-\mathcal{Q}}.
    \end{align*}
    
    \item \textbf{Relaxed Threshold:} Let $\delta \ge 0$ be an integer. Then:
    \begin{align*}
        \m(y \cnd n) &\stackrel{*}{\ge} |T^\delta_n(y)| \cdot 2^{-\mathcal{N}-O(\delta)}, \\
        \m(y \cnd q) &\stackrel{*}{\ge} |T^\delta_q(y)| \cdot 2^{-\mathcal{Q}-O(\delta)}.
    \end{align*}
\end{enumerate}
\end{lemma}

\paragraph{Proof of Item 1.}
We construct a description procedure based on the measure $\m$. Let $c$ be a sufficiently large constant. Discretize the probability space into ``quanta'' of size $\epsilon = 2^{-(\mathcal{N} - c)}$.

Consider an algorithm that simulates the enumeration of the lower semicomputable semimeasure $\m(\cdot \cnd n)$. The algorithm takes an input $r$ (interpreted as an index $0 \le r < 2^{\mathcal{N} - c}$) and searches for the $r$-th ``quantum'' of probability mass.
Specifically, it monitors the enumeration of $\m(\cdot \cnd n)$ and waits for the moment when for some string $y'$, the cumulative value $\m(y' \cnd n)$ reaches at least $k \cdot \epsilon$ for some new integer $k$. The algorithm matches the global index $r$ to this specific event and outputs the pair $(y', k)$.

Since the input $r$ has length $\mathcal{N} - c$, the complexity of the output pair is bounded:
\[ \C_V(y', k) \le |r| + O(1) = \mathcal{N} - c + O(1). \]
By choosing $c$ large enough to cover the $O(1)$ overhead, we ensure $\C_V(y', k) \le \mathcal{N}$. Thus, $k$ (interpreted as a string extension) belongs to the set $T_n(y')$.

Now consider our fixed string $y$. From the length condition $|y| < \mathcal{N}/2 - O(1)$ we have:
\[ \m(y \cnd n) \ge 2^{-(\mathcal{N}-c)} = \epsilon. \]
Since the measure $\m(y \cnd n)$ strictly exceeds $\epsilon$, the number of full quanta allocated to $y$ is strictly positive.
The exact number of such distinct quanta $k$ is $\lfloor \m(y \cnd n) / \epsilon \rfloor$.
Each such quantum $k$ corresponds to a distinct valid extension in $T_n(y)$.
Therefore, $|T_n(y)| \ge \frac{1}{2} \m(y \cnd n) \cdot 2^{\mathcal{N} - c}$, which implies:
\[ \m(y \cnd n) \stackrel{*}{\le} |T_n(y)| \cdot 2^{-\mathcal{N}}. \]

\paragraph{Proof of Item 2.}
Let $K = \mathcal{N} + \delta$. We define a specific semimeasure $P_\delta$ relative to the threshold $K$.
\[ P_\delta(y) = |T^\delta_n(y)| \cdot 2^{-K - 1}. \]
Since the set of pairs with complexity $\le K$ is computably enumerable, $P_\delta(y)$ is lower semicomputable (given $n$ and $\delta$). The sum of $P_\delta(y)$ over all $y$ is bounded by the total number of pairs with complexity $\le K$, normalized by $2^{-K}$:
\[ \sum_y P_\delta(y) \le \sum_{y, s: \C_V(y, s) \le K} 2^{-K-1} \le 2^{K+1} \cdot 2^{-K-1} = 1. \]
Thus, $P_\delta$ is a valid semimeasure relative to the condition $\langle n, \delta \rangle$.
By the maximality property of the universal semimeasure $\m$, we have:
\[ \m(y \cnd n, \delta) \stackrel{*}{\ge} P_\delta(y) = |T^\delta_n(y)| \cdot 2^{-\mathcal{N} - \delta - 1}. \]
Using the Coding Theorem (Theorem~\ref{theorem:coding_theorem}), 
\[ -\log \m(y \cnd n) = \K(y \cnd n) + O(1) \le \K(y \cnd n, \delta) + \K(\delta) + O(1). \]
In terms of measures, this translates to:
\[ \m(y \cnd n) \stackrel{*}{\ge} \m(y \cnd n, \delta) \cdot 2^{-\K(\delta)} \ge \m(y \cnd n, \delta) \cdot 2^{-O(\delta)}. \]
Combining these inequalities and absorbing the multiplicative constant into the $O(\delta)$ term in the exponent, we obtain:
\[ \m(y \cnd n) \ge |T^\delta_n(y)| \cdot 2^{-\mathcal{N}-O(\delta)}. \]

This yields the following immediate consequence:
\begin{corollary}[Relaxed vs. Strict Volume]
\label{cor:relaxed_strict}
For any string $y$ with length $|y| < \mathcal{N}/2 - O(1)$ and integer $\delta \ge 0$:
\begin{align*}
    |T_n(y)| &\ge |T^\delta_n(y)| \cdot 2^{-O(\delta)}, \\
    |T_q(y)| &\ge |T^\delta_q(y)| \cdot 2^{-O(\delta)}.
\end{align*}
\end{corollary}

\begin{corollary}[Conditioning Shift]
\label{cor:conditioning_shift}
Let $y$ and $z$ be strings such that $\C_V(y \cnd z) = \delta$ and 
$|y|, |z| < \mathcal{N}/2 - O(1)$. Then:
\begin{align*}
    |T_n(y)| &\ge |T_n(z)| \cdot 2^{-O(\delta)}, \\
    |T_q(y)| &\ge |T_q(z)| \cdot 2^{-O(\delta)}.
\end{align*}
\end{corollary}

\begin{proof}
We prove the statement for the parameter $n$.
Consider any $s$ such that $\C_V(z, s) \le \mathcal{N}$.
We can compute the pair $(y, s)$ by first extracting $z$ from the description of $(z, s)$, and then applying the conditional program for $y$ given $z$. To concatenate these two plain descriptions, we encode the conditional program in a self-delimiting way,
which increases its length to $\delta + O(\log \delta)$ bits.

Thus, the combined description length is bounded by:
\[ \C_V(y, s) \le \C_V(z, s) + \delta + O(\log \delta) \le \mathcal{N} + O(\delta). \]
Thus, $s \in T^{\delta'}_n(y)$ for $\delta' = O(\delta)$.
Consequently, $|T^{\delta'}_n(y)| \ge |T_n(z)|$.
Applying Corollary~\ref{cor:relaxed_strict} with parameter $\delta'$, we obtain:
\[ |T_n(y)| \ge |T^{\delta'}_n(y)| \cdot 2^{-O(\delta')} \ge |T_n(z)| \cdot 2^{-O(\delta)}. \]
\end{proof}

\subsection{The Construction of \texorpdfstring{$\tilde{q}$}{q-tilde}}

We require a lemma that constructs a specific string $\tilde{q}$ (related to the threshold $q$) with precise complexity properties.

\begin{lemma}[The $\tilde{q}$-Structure Lemma]
\label{lemma:q_tilde}
There exists a deterministic algorithm $\mathcal{B}$ with the following properties.

\textbf{Input:}
\begin{itemize}
    \item An integer $n$.
    \item A binary string $w$.
    \item An integer $q \le n^{\log n}$.
\end{itemize}

\textbf{Output:} A binary string $\tilde{q}$.

\textbf{Properties:} The algorithm $\mathcal{B}$ always halts. Furthermore, if $w$ is the \emph{last} string in the enumeration of the set $\{ y \mid \C_V(y) \le n \}$, then the output $\tilde{q}$ satisfies the following conditions:

\begin{enumerate}
    \item \textbf{Small Length:} The length of the string $\tilde{q}$ is logarithmic in $q$.
    \[ |\tilde{q}| = O(\log q). \]

    \item \textbf{Encodes $q$:} The threshold $q$ is effectively computable from $\tilde{q}$.
    \[ \C_V(q \cnd \tilde{q}) = O(1). \]

    \item \textbf{Complexity Decomposition:}
    Let $\alpha := \K(\tilde{q} \cnd n)$. Then:
    \[
        \K(\tilde{q} \cnd n) = \K\Big( \tilde{q}, \, \alpha, \, \K\big(\tilde{q}, \alpha \cnd q\big) \Bigm| n \Big) + O(1).
    \]
\end{enumerate}
\end{lemma}

We explain the meaning of the Complexity Decomposition condition in the next subsection.

The proof of this lemma is given in Section~\ref{subsec:qstr}.

\subsection{Proof of the Complexity Symmetry Shift}

We now prove the key corollary of Lemma~\ref{lemma:q_tilde}.

\begin{lemma}[Complexity Symmetry Shift]
\label{cor:symmetry}
Assume $\K(n \cnd q) = O(1)$. Let $\tilde{q}$ be a string satisfying the properties of Lemma~\ref{lemma:q_tilde}. Then for any string $x$, the following identity holds up to an $O(1)$ additive term:
\[
    \K(x, \tilde{q}, \K(\tilde{q} \cnd n) \cnd q) - \K(\tilde{q}, \K(\tilde{q} \cnd n) \cnd q) = \K(\tilde{q}, x \cnd n) - \K(\tilde{q} \cnd n).
\]
\end{lemma}

\begin{proof}
For convenience we denote the Left Hand Side (relative to $q$) as \textbf{LHS} and the Right Hand Side (relative to $n$) as \textbf{RHS}.

\subsection*{Part 1: The ``Easy'' Direction ($\text{LHS} \le \text{RHS}$)}

First, we establish that the inequality holds in one direction purely based on information content, regardless of the internal structure of $\tilde{q}$.

\paragraph{Analyzing the RHS.}
We expand the RHS using the conditional chain rule relative to $n$:
\[ \text{RHS} = \K(\tilde{q}, x \cnd n) - \K(\tilde{q} \cnd n) = \K(x \cnd \tilde{q}, \K(\tilde{q} \cnd n), n) + O(1). \]
The condition set for the RHS is effectively:
\[ C_{RHS} = \{ \tilde{q}, \, \K(\tilde{q} \cnd n), \, n \}. \]

\paragraph{Analyzing the LHS.}
Let $Y$ be the pair consisting of $\tilde{q}$ and its complexity relative to $n$:
\[ Y = (\tilde{q}, \, \K(\tilde{q} \cnd n)). \]
The LHS can be rewritten using the chain rule relative to $q$:
\[ \text{LHS} = \K(x, Y \cnd q) - \K(Y \cnd q) = \K(x \cnd Y, \K(Y \cnd q), q) + O(1). \]
Expanding $Y$, the condition set for the LHS is:
\[ C_{LHS} = \{ \tilde{q}, \, \K(\tilde{q} \cnd n), \, \K(Y \cnd q), \, q \}. \]

\paragraph{Comparison.}
Notice that $C_{LHS}$ contains the threshold $q$. Since $\K(n \cnd q) = O(1)$, knowledge of $q$ allows the reconstruction of $n$. Furthermore, $C_{LHS}$ explicitly contains $\tilde{q}$ and $\K(\tilde{q} \cnd n)$, which are the core components of $C_{RHS}$.
Thus, the oracle in the LHS has access to strictly more information than the oracle in the RHS.
Therefore, $\text{LHS} \le \text{RHS} + O(1)$.

\subsection*{Part 2: Deriving the Necessary Structure of $\tilde{q}$}

The reverse inequality, $\text{RHS} \le \text{LHS}$, is not true for an arbitrary string. We show that the structure of $\tilde{q}$ guaranteed by Lemma~\ref{lemma:q_tilde} is sufficient to make this inequality hold.

We introduce a ``carrier'' tuple $T$ corresponding to the decomposition property. We require the following invariant (which is exactly Property 3 of Lemma~\ref{lemma:q_tilde}):
\begin{equation}
\label{eq:T_invariant}
\K(T \cnd n) = \K(\tilde{q} \cnd n) + O(1).
\end{equation}

\paragraph{Step A: Upper Bounding the RHS.}
We upper bound the RHS by replacing $\tilde{q}$ with the more informative tuple $T$. Since $T$ contains $\tilde{q}$, the pair $(\tilde{q}, x)$ is trivially recoverable from $(T, x)$. Using the invariant \eqref{eq:T_invariant}, we get:
\begin{align*}
    \text{RHS} &= \K(\tilde{q}, x \cnd n) - \K(\tilde{q} \cnd n) \\
               &\le \K(T, x \cnd n) - \K(T \cnd n) + O(1).
\end{align*}
Now we apply the chain rule to the term $\K(T, x \cnd n)$:
\[ \K(T, x \cnd n) = \K(T \cnd n) + \K(x \cnd T, \K(T \cnd n), n) + O(1). \]
Substituting this back, the term $\K(T \cnd n)$ cancels out:
\begin{equation}
\label{eq:rhs_bound}
    \text{RHS} \le \K(x \cnd T, \K(T \cnd n), n) + O(1).
\end{equation}

\paragraph{Step B: Collapsing the LHS.}
Recall the expansion of the LHS:
\[ \text{LHS} = \K(x \cnd \tilde{q}, \, \K(\tilde{q} \cnd n), \, \K(Y \cnd q), \, q) + O(1), \]
where
\[ Y = (\tilde{q}, \K(\tilde{q} \cnd n)). \]

To enable the bound $\text{RHS} \le \text{LHS}$, we define $T$ to explicitly capture all components of the LHS condition:
\[ T = \Big( \tilde{q}, \, \underbrace{\K(\tilde{q} \cnd n)}_{\alpha}, \, \underbrace{\K\big( (\tilde{q}, \alpha) \cnd q \big)}_{\beta} \Big). \]
With this specific definition:
\begin{enumerate}
    \item $T$ contains $\tilde{q}$ and $\alpha = \K(\tilde{q} \cnd n)$.
    \item $T$ contains $\beta$, which is the complexity of the first two components relative to $q$. Thus, the term $\K(Y \cnd q)$ in the LHS condition is computable from $T$.
    \item Since $T$ contains $\tilde{q}$ (and $\C_V(q \cnd \tilde{q}) = O(1)$), $q$ is computable from $T$.
\end{enumerate}
Therefore, given $T$, all other conditions in the LHS are redundant:
\begin{equation}
\label{eq:lhs_bound}
    \text{LHS} = \K(x \cnd T) + O(1).
\end{equation}

\paragraph{Step C: Final Comparison.}
We now compare the derived bounds.
\begin{itemize}
    \item From \eqref{eq:rhs_bound}: $\text{RHS} \le \K(x \cnd T, \K(T \cnd n), n) + O(1)$.
    \item From \eqref{eq:lhs_bound}: $\text{LHS} = \K(x \cnd T)+ O(1)$.
\end{itemize}
The condition in the RHS bound includes $T$. Since $T$ allows the computation of $\K(T \cnd n)$ (via $\alpha$) and $n$ (via $q$), the condition sets are informationally equivalent (up to $O(1)$), ensuring the inequality holds.
Combining results:
\[ \text{RHS} \le \K(x \cnd T, \dots) + O(1) \le \K(x \cnd T) + O(1) = \text{LHS} + O(1). \]
Thus, the Corollary is proved.
\end{proof}

\subsection{Computability: The Oracle Power of \texorpdfstring{$w$}{w}}
\label{subsec:computability}

The definitions of the algorithms in this section (both the threshold computation $\mathcal{A}$ and the structural search $\mathcal{B}$) rely on the prefix complexity function $\K(\cdot)$, which is generally uncomputable. However, these algorithms receive as input the string $w$, defined as the \emph{last} string generated in the canonical enumeration of the set $\{ y \mid \C_V(y) \le n \}$.

We show that $w$ acts as a powerful oracle. Let $T_w$ denote the number of computation steps performed by the optimal machine $V$ until it outputs the string $w$. Since $w$ marks the completion of the enumeration for the threshold $n$, the time $T_w$ acts as a universal cutoff for all computations described by simple programs.

\begin{proposition}[Computability for Simple Objects]
\label{prop:computability}
For any strings $a$ and $b$, if their lengths satisfy
\[ |a| + |b| < n/2, \]
then the value $\K(a \cnd b)$ is computable given inputs $n, a, b$ and $w$.
\end{proposition}

\begin{proof}
First, we determine $T_w$ by running $V$ until it outputs $w$.
The following procedure computes $\K(a \cnd b)$:
\begin{enumerate}
    \item Simulate the universal prefix decompressor on input $(p, b)$ for all candidate programs $p$ with length up to $|a| + O(1)$, running each for at most $T_w$ steps.
    \item Collect all programs that halt and output $a$.
    \item Output the length of the shortest such program.
\end{enumerate}

\paragraph{Correctness.}
Let $p^*$ be a shortest program for the prefix decompressor such that it outputs $a$ given $b$. We must verify that this computation halts within $T_w$ steps.

Consider the description for the plain machine $V$ corresponding to the execution of $p^*$ on $b$. The length of this description is bounded by $|p^*| + |b| + O(1)$. Since $|p^*| \le |a| + O(1)$, the total length is bounded by $|a| + |b| + O(1)$.
By the hypothesis $|a| + |b| < n/2$, this length is strictly less than $n$ (for sufficiently large $n$).

Since $w$ is the last string generated by $V$ with complexity $\le n$, any computation corresponding to a shorter description must complete within $T_w$ steps. Thus, the simulation of the prefix decompressor on $p^*$ and $b$ completes within the time bound.
\end{proof}

\subsection{Proof of the \texorpdfstring{$\tilde{q}$}{q-tilde} Structure Lemma}
\label{subsec:qstr}
\begin{proof}
We search for $\tilde{q}$ among pairs of the form $T_l = (q, l)$, where $l$ is an integer in the range $0 \le l \le q$.
For any such candidate $T_l$, we define its structural complexity components relative to the scales $n$ and $q$:
\begin{align*}
    \alpha(l) &= \K(T_l \cnd n), \\
    \beta(l) &= \K(T_l, \alpha(l) \cnd q).
\end{align*}
Our goal is to find an index $l^*$ such that $\beta(l^*) = l^* + O(1)$.
Note that for every $l$ it holds that $\beta(l+1) = \beta(l) + O(1)$.

\begin{lemma}
\label{lemma:fixed_point}
There exists an integer $l^* \in [0, q-1]$ such that:
\[ \beta(l^*) = l^* + O(1). \]
\end{lemma}

\begin{proof}
Consider the set of indices where the complexity exceeds the index value:
\[ S = \{ l \in [0, q] \mid \beta(l) \ge l \}. \]

\paragraph{1. Boundary Conditions.}
\begin{itemize}
    \item \textbf{Lower Bound ($l=0$):} The value $\beta(0) = \K((q, 0), \dots \cnd q)$ is non-negative. Thus, $\beta(0) \ge 0$, so $0 \in S$, and the set $S$ is non-empty.
    \item \textbf{Upper Bound ($l=q$):} The value $\beta(q) = \K((q, q), \dots \cnd q) = O(\log q)$. Since $\beta(q) = O(\log q)$, for all sufficiently large $q$ we strictly have $\beta(q) < q$. Thus, $q \notin S$.
\end{itemize}

\paragraph{2. The Max Element.}
Define:
\[ l^* = \max S. \]
By the definition of the maximum:
\begin{enumerate}
    \item Since $l^* \in S$, we have $\beta(l^*) \ge l^*$.
    \item Since $l^*$ is the maximum, the next integer $l^*+1$ is not in $S$. Therefore, $\beta(l^*+1) < l^* + 1$.
\end{enumerate}
Since $\beta(l^*+1) = \beta(l^*) + O(1)$, we have:
\[ l^* \le \beta(l^*) \le \beta(l^*+1) + O(1) < l^* + 1 + O(1). \]
This implies $\beta(l^*) = l^* + O(1)$.
\end{proof}

We define our target string as $\tilde{q} = T_{l^*} = (q, l^*)$.

\subsubsection*{Verification of Properties}

We now verify that this choice of $\tilde{q}$ satisfies the three conditions of the Lemma.

\paragraph{1. Small Length.}
Since $\tilde{q} = (q, l^*)$ and $l^* \le q$, the string is composed of binary representations of two numbers bounded by $q$.
\[ |\tilde{q}| = O(\log q). \]

\paragraph{2. Encodes $q$.}
Since $\tilde{q} = (q, l^*)$, the value $q$ is explicitly the first element of the pair.
\[ \C_V(q \cnd \tilde{q}) = O(1). \]

\paragraph{3. Complexity Decomposition.}
We must prove:
\[ \K(\tilde{q} \cnd n) = \K\Big( \tilde{q}, \, \K(\tilde{q} \cnd n), \, \K\big(\tilde{q}, \K(\tilde{q} \cnd n) \cnd q\big) \Bigm| n \Big) + O(1). \]
Let $\alpha = \K(\tilde{q} \cnd n)$ and $\beta = \beta(l^*) = \K(\tilde{q}, \alpha \cnd q)$. The RHS becomes $\K(\tilde{q}, \alpha, \beta \cnd n)$.

By Lemma~\ref{lemma:fixed_point}, we have $|\beta - l^*| \le O(1)$. Since $\tilde{q} = (q, l^*)$, the value $l^*$ is contained in $\tilde{q}$. Therefore, given $\tilde{q}$, the value $\beta$ can be described using only $O(1)$ bits.
This implies that adding $\beta$ to the condition provides no new information up to a constant:
\[ \K(\tilde{q}, \alpha, \beta \cnd n) = \K(\tilde{q}, \alpha \cnd n) + O(1). \]
Now we apply Proposition~\ref{prop:self_sufficiency} (Complexity Self-Sufficiency), which states $\K(x, \K(x \cnd y) \cnd y) = \K(x \cnd y) + O(1)$. With $x=\tilde{q}$ and $y=n$, this yields:
\[ \K(\tilde{q}, \alpha \cnd n) = \K(\tilde{q} \cnd n) + O(1). \]
Combining these steps proves the decomposition identity.

\subsubsection*{Computability of Algorithm $\mathcal{B}$}

Algorithm $\mathcal{B}$ operates by iterating $l$ from $0$ to $q$, computing the values $\alpha(l)$ and $\beta(l)$ at each step, and halting when the condition of Lemma~\ref{lemma:fixed_point} is met.

The validity of this procedure relies on the computability of the prefix complexities involved. 
For the computation of $\beta(l) = \K(T_l, \alpha(l) \cnd q)$, the condition is $b = q$ and the target is $a = (T_l, \alpha(l))$. The lengths are $|T_l| = O(\log q)$ and $|\alpha(l)| = O(\log \log q)$, so the total length of the input parameters is bounded by $O(\log q)$. By the premise of the Lemma, $q \le n^{\log n}$, meaning $\log q \le (\log n)^2$. For sufficiently large $n$, we have $O((\log n)^2) \ll n/2$.

Therefore, the length requirement $|a| + |b| < n/2$ of \textbf{Proposition~\ref{prop:computability}} is strictly satisfied for both $\alpha(l)$ and $\beta(l)$. This allows the algorithm to compute the exact values of $\alpha(l)$ and $\beta(l)$ by simulating the universal machine with the time bound $T_w$ derived from the input $w$.
Since the existence of $l^*$ is guaranteed, the algorithm always halts.
\end{proof}

\subsection{Volume Decomposition}

We decompose the set of extensions based on the moment the string $w$ appears. Let $w$ be the last string in the canonical enumeration of strings with complexity at most $n$.

We define the \textbf{``Past'' sets} $St^w$ as the sets of extensions discovered by the universal machine prior to the appearance of $w$. Specifically, for any string $y$:
\begin{align*}
    St^{w}_q(y) &= \{ s : \C_V(y, s) \le \mathcal{Q} \text{ and } (y, s) \text{ is enumerated before } w \}, \\
    St^{w}_n(y) &= \{ s : \C_V(y, s) \le \mathcal{N} \text{ and } (y, s) \text{ is enumerated before } w \}.
\end{align*}
\begin{lemma}[Volume Decomposition]
\label{lemma:volume_decomposition}
Let $\tilde{q}$ and $\alpha$ be values defined in Lemma~\ref{lemma:q_tilde}. Then:

\begin{enumerate}
    \item $|T_q(x)| \stackrel{*}{=} |\mathrm{St}^{w}_q(x)| + |T_q(x, \tilde{q}, \alpha)|$.
    \item $|T_n(\tilde{q}, x)| \stackrel{*}{=} |\mathrm{St}^{w}_n(\tilde{q}, x)| + |T_n(x, q, w)|.$
\end{enumerate}
\end{lemma}

\begin{proof}
We prove the equality up to a multiplicative constant ($\stackrel{*}{=}$) by establishing both the lower ($\stackrel{*}{\ge}$) and upper ($\stackrel{*}{\le}$) bounds for each item.

\paragraph{Lower Bounds ($\stackrel{*}{\ge}$):}
By definition, the ``Past'' sets are subsets of the total extension sets, so $|T_q(x)| \ge |\mathrm{St}^{w}_q(x)|$ and $|T_n(\tilde{q}, x)| \ge |\mathrm{St}^{w}_n(\tilde{q}, x)|$.

For the remaining terms, we apply the Conditioning Shift (Corollary~\ref{cor:conditioning_shift}):
\begin{itemize}
    \item For Item 1: The tuple $(x, \tilde{q}, \alpha)$ trivially determines $x$ (so $\C_V(x \cnd x, \tilde{q}, \alpha) = O(1)$). Thus, dropping the extra conditions only decreases complexity, giving $|T_q(x)| \stackrel{*}{\ge} |T_q(x, \tilde{q}, \alpha)|$.
    \item For Item 2: The tuple $(\tilde{q}, x)$ can be algorithmically reconstructed from $(x, q, w)$ (since $n$ is determined by $q$, and $\tilde{q}$ is computed from $q$ and $w$ via Lemma~\ref{lemma:q_tilde}). Thus, $\C_V(\tilde{q}, x \cnd x, q, w) = O(1)$, which implies $|T_n(\tilde{q}, x)| \stackrel{*}{\ge} |T_n(x, q, w)|$.
\end{itemize}
Combining these (since the sets are disjoint or one dominates), the lower bounds hold.

\paragraph{Upper Bounds ($\stackrel{*}{\le}$):}
We must bound the number of ``Future'' extensions—those enumerated \textit{after} $w$.

\textbf{Proof for Item 1:}
Consider any string $s \in T_q(x) \setminus \mathrm{St}^{w}_q(x)$.
Since $(x, s)$ has complexity $\le \mathcal{Q}$ but is enumerated after $w$, any optimal plain program $p$ outputting $(x, s)$ has length $|p| \le \mathcal{Q}$ and halts after $w$ has appeared.

To upper bound the number of such extensions, we construct a new program that outputs the tuple $(x, \tilde{q}, \alpha, s)$. To do this, our algorithm needs to know the threshold $\mathcal{Q}$ (and thus $q$). We consider two cases based on the length of $p$:

\textit{Case 1: The simple case ($|p| = \mathcal{Q}$).}
Assume the length of $p$ is exactly $\mathcal{Q}$. The algorithm can deduce $\mathcal{Q}$ simply by measuring the length of its input $p$. By running $p$, we observe its halting state. Since $w$ is the absolute last string of complexity $\le n$ to ever be enumerated by $V$, and $p$ halts after $w$ has appeared, the last such string enumerated up to the halting time of $p$ is guaranteed to be exactly $w$.
Having identified $w$ and knowing $q$, we can compute $\tilde{q}$ using the algorithm from Lemma~\ref{lemma:q_tilde}. Furthermore, since $\tilde{q}$ is simple relative to $n$, we compute $\alpha = \K(\tilde{q} \cnd n)$ using Proposition~\ref{prop:computability}. The algorithm then outputs $(x, \tilde{q}, \alpha, s)$. The length of this new description is exactly $|p| + O(1) = \mathcal{Q} + O(1)$.

\textit{Case 2: The general case ($|p| < \mathcal{Q}$).}
If $|p| = L < \mathcal{Q}$, let $\Delta = \mathcal{Q} - L$ be the difference. We can provide $\Delta$ to our algorithm using a self-delimiting prefix code, which takes $O(\log \Delta)$ bits, followed by the raw plain program $p$. The algorithm decodes $\Delta$, reads the remaining $L$ bits as $p$, calculates the threshold $\mathcal{Q} = L + \Delta$, and proceeds exactly as in Case 1. The total length of this description is:
\[ L + O(\log \Delta) + O(1) \le L + \Delta + O(1) = \mathcal{Q} + O(1). \]

In both cases, the description length is bounded by $\mathcal{Q} + O(1)$. Therefore, every such ``Future'' extension $s$ belongs to the relaxed set $T^{O(1)}_q(x, \tilde{q}, \alpha)$.
The total number of such extensions is bounded by $|T^{O(1)}_q(x, \tilde{q}, \alpha)|$. Applying the Relaxed vs. Strict Volume property (Corollary~\ref{cor:relaxed_strict}), we get:
\[ |T^{O(1)}_q(x, \tilde{q}, \alpha)| \stackrel{*}{\le} |T_q(x, \tilde{q}, \alpha)|. \]
Summing the ``Past'' and ``Future'' components yields the upper bound for Item 1.

\textbf{Proof for Item 2:}
Consider any string $s \in T_n(\tilde{q}, x) \setminus \mathrm{St}^{w}_n(\tilde{q}, x)$.
Similarly, the optimal program of length $\le \mathcal{N}$ outputting $(\tilde{q}, x, s)$ halts after $w$ has appeared. During its decompression, we can identify $w$ exactly as described above.
Since $\tilde{q}$ structurally encodes $q$, we can trivially extract $q$ from $\tilde{q}$ in $O(1)$ steps.
We can then output the tuple $(x, q, w, s)$.
The complexity of this pair is bounded by $\mathcal{N} + O(1)$. Therefore, all such ``Future'' extensions belong to $T^{O(1)}_n(x, q, w)$.
By Corollary~\ref{cor:relaxed_strict}, the number of these extensions is bounded by $|T_n(x, q, w)| \cdot 2^{O(1)}$. Summing the Past and Future components yields the required upper bound for Item 2.
\end{proof}

\subsection{Summary of Structural Properties}
\label{subsec:structure_summary}

In this subsection, we consolidate the results established in the previous lemmas. These properties serve as the foundation for the algorithm $\mathcal{A}$ and the final counting argument.

\begin{lemma}[Structural Volume Lemma]
\label{lemma:structural_summary}
Let $n, x, q$ satisfy the assumptions of the Key Lemma~\ref{lemma:unified_gap}. Let $w$ be the last string in the canonical enumeration of strings with complexity at most $n$.
Let $\tilde{q}$ be the structural witness derived from $n, w, q$ via the $\tilde{q}$-Structure Lemma \ref{lemma:q_tilde}, and let $\alpha := \K(\tilde{q} \cnd n)$.

Then the following properties hold:

\begin{enumerate}
    \item \textbf{Low Complexity Implies Large Volume:}
    \[ \C_V(x) \le q \implies |T_q(x)| \stackrel{*}{\ge} 2^{-q} \cdot 2^{\mathcal{Q}}. \]

    \item \textbf{Volume Decomposition (Target Level):}
    \[ |T_q(x)| \stackrel{*}{=} |\mathrm{St}^{w}_q(x)| + |T_q(x, \tilde{q}, \alpha)|. \]

    \item \textbf{Log-Volume of the Remainder:}
    \[
    -\log |T_q(x, \tilde{q}, \alpha)| = \K(\tilde{q}, \alpha \cnd q) + \K(\tilde{q}, x \cnd n) - \K(\tilde{q} \cnd n) - \mathcal{Q} + O(1).
    \]

    \item \textbf{Computability:}
    The terms $\K(\tilde{q}, \alpha \cnd q)$ and $\K(\tilde{q} \cnd n)$ can be effectively computed given $n, x, q, w$.

    \item \textbf{Complexity Decomposition (Resource Level):}
    The complexity at level $n$ is determined by the dominant volume component:
    \[
    \K(\tilde{q}, x \cnd n) = \mathcal{N} - \log \max\left( |\mathrm{St}^{w}_n(\tilde{q}, x)|, |T_n(x, q, w)| \right) + O(1).
    \]
\end{enumerate}
\end{lemma}

\begin{proof}
\textbf{Item 1.}
By Corollary~\ref{cor:equivalence_bounds}, $\C_V(x) \le q$ implies $\K(x \cnd q) \le q + O(1)$.
By the Coding Theorem \ref{theorem:coding_theorem}, $\K(x \cnd q) = -\log \m(x \cnd q) + O(1)$. Thus:
\[ -\log \m(x \cnd q) \le q + O(1) \implies \m(x \cnd q) \stackrel{*}{\ge} 2^{-q}. \]
Finally, by the Measure Density Lemma \ref{lemma:measure_density} (applied to threshold $\mathcal{Q}$):
\[ |T_q(x)| \cdot 2^{-\mathcal{Q}} \stackrel{*}{\ge} \m(x \cnd q) \stackrel{*}{\ge} 2^{-q}. \]
Multiplying by $2^{\mathcal{Q}}$ yields the result.

\textbf{Item 2} follows directly from the definition of the decomposition sets in Lemma \ref{lemma:volume_decomposition}.

\textbf{Item 3.} We first express the volume of the set in terms of its complexity. By the Measure Density Lemma \ref{lemma:measure_density} (applied to threshold $\mathcal{Q}$) and the Coding Theorem:
\[ \log |T_q(x, \tilde{q}, \alpha)| = \mathcal{Q} - \K(x, \tilde{q}, \alpha \cnd q) + O(1). \]
Equivalently:
\[ -\log |T_q(x, \tilde{q}, \alpha)| = \K(x, \tilde{q}, \alpha \cnd q) - \mathcal{Q} + O(1). \]
Next, we expand the complexity term using the Complexity Symmetry Shift (Lemma \ref{cor:symmetry}):
\[ \K(x, \tilde{q}, \alpha \cnd q) - \K(\tilde{q}, \alpha \cnd q) = \K(\tilde{q}, x \cnd n) - \K(\tilde{q} \cnd n). \]
Solving for $\K(x, \tilde{q}, \alpha \cnd q)$ and substituting into the previous expression yields:
\[ -\log |T_q(x, \tilde{q}, \alpha)| = \left( \K(\tilde{q}, \alpha \cnd q) + \K(\tilde{q}, x \cnd n) - \K(\tilde{q} \cnd n) \right) - \mathcal{Q} + O(1). \]

\textbf{Item 4} follows from Proposition \ref{prop:computability}, since the lengths of $\tilde{q}$, $\alpha$, and $q$ satisfy the condition of that Proposition.

\textbf{Item 5.}
We apply the Volume Decomposition Lemma \ref{lemma:volume_decomposition} to level $n$ (Item 2). The total volume decomposes (up to multiplicative constants) as:
\[ |T_n(\tilde{q}, x)| \stackrel{*}{=} |\mathrm{St}^{w}_n(\tilde{q}, x)| + |T_n(x, q, w)|. \]
Taking the logarithm:
\[ \log |T_n(\tilde{q}, x)| = \log \left( |\mathrm{St}^{w}_n(\tilde{q}, x)| + |T_n(x, q, w)| \right) + O(1). \]
Since $\log(A+B) = \max(\log A, \log B) + O(1)$, we have:
\[ \log |T_n(\tilde{q}, x)| = \max \left( \log |\mathrm{St}^{w}_n(\tilde{q}, x)|, \log |T_n(x, q, w)| \right) + O(1). \]
By the Measure Density Lemma (applied to threshold $\mathcal{N}$), the complexity relates to the total volume as:
\[ \K(\tilde{q}, x \cnd n) = \mathcal{N} - \log |T_n(\tilde{q}, x)| + O(1). \]
Substituting the max-expression yields the required formula.
\end{proof}

\subsection{Proof of the Key Lemma}

\subsubsection*{The Algorithm $\mathcal{A}$}

We introduce a sufficiently large constant $\lambda$ (to cover precision gaps in the structural lemmas). The algorithm receives inputs $n, w, x, q$ and operates as follows:

\begin{enumerate}
    \item \textbf{Check Triviality at $q$ (High Past Volume):} 
    Compute the volume of extensions of $x$ at scale $q$ accumulated in the structure set:
    \[ V_{check} = |\mathrm{St}^{w}_q(x)|. \]
    If $V_{check} \ge 2^{-q - \lambda} \cdot 2^{\mathcal{Q}}$, output $B=0$ and terminate.

    \item \textbf{Compute Structural Correction:} 
    Calculate the correction factor $\Delta$ using the computable terms for $\tilde{q}$ (as guaranteed by Item 4 of Lemma~\ref{lemma:structural_summary}):
    \[ \Delta = \K(\tilde{q}, \alpha \cnd q) - \K(\tilde{q} \cnd n). \]

    \item \textbf{Compute Structure Volume at $n$:} 
    Compute the volume of extensions of the pair $(\tilde{q}, x)$ at scale $n$ accumulated so far:
    \[ V_{str} = |\mathrm{St}^{w}_n(\tilde{q}, x)|. \]

    \item \textbf{Output Threshold:} 
    Define the target count:
    \[ \mathrm{Num} = 2^{\mathcal{N} - q + \Delta - \lambda}. \]
    If $V_{str} \ge \mathrm{Num}$, output $B=0$ (the structure is already ``heavy'' enough).
    Otherwise, output $B = \mathrm{Num}$.
\end{enumerate}

\subsubsection*{Proof of Property 1: Low Complexity implies High Count}

We must show that if $\C_V(x) \le q$, then $|T_n(x, q, w)| \ge B$.

\paragraph{Step 1: Complexity to Volume.}
Assume $\C_V(x) \le q$. By \textbf{Item 1} of Lemma~\ref{lemma:structural_summary} (Low Complexity Implies Large Volume), the total volume of extensions at scale $q$ is large:
\[ |T_q(x)| \stackrel{*}{\ge} 2^{-q} \cdot 2^{\mathcal{Q}}. \]

\paragraph{Step 2: Decomposition at Scale $q$.}
By \textbf{Item 2} (Volume Decomposition), the total set partitions into ``Past'' and ``Future'':
\[ |T_q(x)| \stackrel{*}{=} |\mathrm{St}^{w}_q(x)| + |T_q(x, \tilde{q}, \alpha)|. \]
If the ``Past'' term $|\mathrm{St}^{w}_q(x)|$ is large enough to satisfy the condition in Step 1 of Algorithm $\mathcal{A}$, then $B=0$, and the condition $|T_n| \ge 0$ holds trivially. 
Otherwise, since the constant $\lambda$ in Algorithm $\mathcal{A}$ is chosen sufficiently large to strictly absorb any implicit multiplicative constants from the $\stackrel{*}{\ge}$ notation, the ``Future'' term must dominate the volume:
\[ |T_q(x, \tilde{q}, \alpha)| \stackrel{*}{\ge} 2^{-q} \cdot 2^{\mathcal{Q}}. \]

\paragraph{Step 3: Symmetry Application.}
We use \textbf{Item 3} of Lemma~\ref{lemma:structural_summary} to convert the volume bound at scale $q$ into a complexity bound at scale $n$:
\[ -\log |T_q(x, \tilde{q}, \alpha)| = \Delta + \K(\tilde{q}, x \cnd n) - \mathcal{Q} + O(1). \]
Substituting the lower bound from Step 2 (which implies $-\log |T_q| \le q - \mathcal{Q} + O(1)$):
\[ q - \mathcal{Q} \ge \Delta + \K(\tilde{q}, x \cnd n) - \mathcal{Q} - O(1). \]
Simplifying $\mathcal{Q}$:
\begin{equation}
\label{eq:compl_bound_final}
\K(\tilde{q}, x \cnd n) \le q - \Delta + O(1).
\end{equation}

\paragraph{Step 4: Decomposition at Scale $n$.}
We map this complexity bound back to volumes at level $n$ using \textbf{Item 5} of Lemma~\ref{lemma:structural_summary}. Recall that the volume is determined by the maximum of its components. Since $\max(A, B) \le A+B \le 2\max(A, B)$, we can safely replace the maximum with a sum by absorbing a $1$-bit difference into the $O(1)$ term:
\[ \K(\tilde{q}, x \cnd n) = \mathcal{N} - \log \left( V_{str} + |T_n(x, q, w)| \right) + O(1). \]
Substituting this into inequality~\eqref{eq:compl_bound_final}:
\[ \mathcal{N} - \log \left( V_{str} + |T_n(x, q, w)| \right) \le q - \Delta + O(1). \]
Rearranging to solve for the total volume:
\[ \log \left( V_{str} + |T_n(x, q, w)| \right) \ge \mathcal{N} - q + \Delta - O(1). \]

\paragraph{Step 5: Final Comparison.}
Exponentiating the result from Step 4, we get:
\begin{equation}
\label{eq:total_vol_bound}
V_{str} + |T_n(x, q, w)| \ge 2^{\mathcal{N} - q + \Delta - O(1)}.
\end{equation}
Recall that $\mathrm{Num} = 2^{\mathcal{N} - q + \Delta - \lambda}$. The bound~\eqref{eq:total_vol_bound} can be rewritten as:
\[ V_{str} + |T_n(x, q, w)| \ge \mathrm{Num} \cdot 2^{\lambda - O(1)}. \]
Let $C = 2^{\lambda - O(1)}$. Since $\lambda$ is chosen sufficiently large, $C \ge 2$. Thus:
\[ V_{str} + |T_n(x, q, w)| \ge 2 \cdot \mathrm{Num}. \]

We now analyze the two cases of the algorithm's output:
\begin{itemize}
    \item \textbf{Case A:} $V_{str} \ge \mathrm{Num}$.
    The algorithm outputs $B=0$. Since $|T_n(x, q, w)| \ge 0$, the condition holds.

    \item \textbf{Case B:} $V_{str} < \mathrm{Num}$.
    The algorithm outputs $B = \mathrm{Num}$.
    From our derived inequality:
    \[ |T_n(x, q, w)| \ge 2 \cdot \mathrm{Num} - V_{str}. \]
    Since $V_{str} < \mathrm{Num}$, we have:
    \[ |T_n(x, q, w)| > 2 \cdot \mathrm{Num} - \mathrm{Num} = \mathrm{Num} \ge B. \]
\end{itemize}
Thus, in all cases, $|T_n(x, q, w)| \ge B$.

\subsubsection*{Proof of Property 2: High Count implies Low Complexity}

We must show that if the relaxed extension count is large, i.e.,
\[ |T^{\sqrt{d}}_n (x, q, w)| \ge B, \]
then the complexity is bounded by:
\[ \C_V(x) \le q + \frac{d}{2}. \]

\paragraph{Strategy.}
To prove that $\C_V(x) \le q + d/2$, it suffices to establish the following lower bound on the volume at scale $q$:
\begin{equation}
\label{goal}
|T_q(x)| \ge 2^{\mathcal{Q} - q - O(\sqrt{d})}.
\end{equation}
Indeed, applying the \textbf{Measure Density Lemma} (Lemma~\ref{lemma:measure_density}) and the \textbf{Coding Theorem} (Theorem~\ref{theorem:coding_theorem}), this volume bound implies:
\[ \K(x \cnd q) \le q + O(\sqrt{d}). \]
Thus:
\[ \K(x \cnd q + d/3) \le q + O(\sqrt{d}). \]
For sufficiently large $d$, the term $O(\sqrt{d})$ is strictly less than $d/3$. Therefore:
\[ \K(x \cnd q + d/3) \le q + d/3. \]
By Corollary~\ref{cor:equivalence_bounds}, this implies:
\[ \C_V(x) \le q + d/3 + O(1). \]
Finally, for sufficiently large $d$, we obtain the required bound 
\[ \C_V(x) \le q + d/2. \]

We consider the execution path of Algorithm $\mathcal{A}$ that determined the value $B$.

\paragraph{Case 1: Trivial Termination at Step 1 (High Past Volume at $q$).}
Suppose the algorithm terminated at Step 1. This implies the check passed:
\[ |\mathrm{St}^{w}_q(x)| \ge 2^{-q - \lambda} \cdot 2^{\mathcal{Q}}. \]
Since the structure set is a subset of the total set ($\mathrm{St}^{w}_q(x) \subseteq T_q(x)$), we immediately have:
\[ |T_q(x)| \ge |\mathrm{St}^{w}_q(x)| \ge 2^{\mathcal{Q} - q - \lambda}. \]
Since $\lambda$ is a constant and $d$ is sufficiently large, condition (\ref{goal}) is satisfied.

\paragraph{Case 2: Trivial Termination at Step 4 (High Structure Volume).}
Suppose the algorithm terminated at Step 4 because the condition $|\mathrm{St}^{w}_n(\tilde{q}, x)| \ge \mathrm{Num}$ was met.
In this case $B=0$. Although the condition $|T^{\sqrt{d}}_n| \ge B$ holds trivially, the state of the algorithm (large structure volume) is sufficient to upper bound the complexity $\K(\tilde{q}, x \cnd n)$.

We start with the condition:
\[ |\mathrm{St}^{w}_n(\tilde{q}, x)| \ge \mathrm{Num} = 2^{\mathcal{N} - q + \Delta - \lambda}. \]
We apply \textbf{Item 5} of Lemma \ref{lemma:structural_summary} (Complexity Decomposition). Since the maximum is at least the first term:
\[ \K(\tilde{q}, x \cnd n) = \mathcal{N} - \log \max\left( |\mathrm{St}^{w}_n(\tilde{q}, x)|, |T_n(x, q, w)| \right) + O(1). \]
\[ \K(\tilde{q}, x \cnd n) \le \mathcal{N} - \log |\mathrm{St}^{w}_n(\tilde{q}, x)| + O(1). \]
Substituting the lower bound $|\mathrm{St}^{w}_n(\tilde{q}, x)| \ge \mathrm{Num}$:
\[ \K(\tilde{q}, x \cnd n) \le \mathcal{N} - \log(\mathrm{Num}) + O(1). \]
Substituting $\mathrm{Num} = 2^{\mathcal{N} - q + \Delta - \lambda}$:
\begin{equation}
\label{eq:compl_upper_case2}
\K(\tilde{q}, x \cnd n) \le q - \Delta + O(\lambda).
\end{equation}
(Note that $O(\lambda) \subset O(\sqrt{d})$ for large $d$).

Now we transfer this bound to the volume at scale $q$. We use \textbf{Item 3} of Lemma \ref{lemma:structural_summary}:
\[ -\log |T_q(x, \tilde{q}, \alpha)| = \Delta + \K(\tilde{q}, x \cnd n) - \mathcal{Q} + O(1). \]
Substituting the upper bound from (\ref{eq:compl_upper_case2}):
\begin{align*}
-\log |T_q(x, \tilde{q}, \alpha)| &\le q - \mathcal{Q} + O(\sqrt{d}).
\end{align*}
So:
\[ |T_q(x, \tilde{q}, \alpha)| \ge 2^{\mathcal{Q} - q - O(\sqrt{d})}. \]
By using Corollary~\ref{cor:conditioning_shift} (dropping the simple conditions $\tilde{q}, \alpha$), we conclude:
\[ |T_q(x)| \ge |T_q(x, \tilde{q}, \alpha)| \cdot 2^{-O(1)} \ge 2^{\mathcal{Q} - q - O(\sqrt{d})}. \]

\paragraph{Case 3: Main Case (Step 4, Non-trivial).}
Suppose the algorithm computed the non-trivial threshold $B$ in the `Otherwise` branch of Step 4. This implies that $B = \mathrm{Num}$.
Assume the premise of Property 2 holds: $|T^{\sqrt{d}}_n (x, q, w)| \ge B$. We derive the same complexity bound as in Case 2, but relying on the volume of extensions.

First, by Lemma~\ref{lemma:measure_density} (Item 2), the strict count is related to the relaxed count as:
\[ |T_n(x, q, w)| \ge |T^{\sqrt{d}}_n (x, q, w)| \cdot 2^{-O(\sqrt{d})} \ge \mathrm{Num} \cdot 2^{-O(\sqrt{d})}. \]
Next, we again apply \textbf{Item 5} of Lemma \ref{lemma:structural_summary}. The complexity is bounded by the logarithm of the maximum volume. In this case, we use the fact that the maximum is at least the second term:
\[ \K(\tilde{q}, x \cnd n) = \mathcal{N} - \log \max\left( |\mathrm{St}^{w}_n(\tilde{q}, x)|, |T_n(x, q, w)| \right) + O(1). \]
\[ \K(\tilde{q}, x \cnd n) \le \mathcal{N} - \log |T_n(x, q, w)| + O(1). \]
Substituting the strict volume bound derived above:
\[ \K(\tilde{q}, x \cnd n) \le \mathcal{N} - \log(\mathrm{Num} \cdot 2^{-O(\sqrt{d})}) + O(1). \]
Substituting $\mathrm{Num}$:
\[ \K(\tilde{q}, x \cnd n) \le \mathcal{N} - (\mathcal{N} - q + \Delta - \lambda) + O(\sqrt{d}). \]
Simplifying:
\[ \K(\tilde{q}, x \cnd n) \le q - \Delta + O(\sqrt{d}).\]

The rest of the derivation (transferring to scale $q$) is identical to Case 2, leading to the required bound on $|T_q(x)|$.

\appendix

\section{Omitted Proofs}
\label{sec:appendix_proofs}

\subsection*{Proof of Lemma \ref{lem:impl}}

\begin{proof}[Proof of Lemma \ref{lem:impl}]
    Assume $H \in \PP^{F_{\C_U}}$.
    Consider the predicate $Q(m, k)$: ``Are there at least $k$ distinct strings $x$ such that $\C_U(x) \le m$?''
    
    This predicate is computably enumerable (c.e.), as one can enumerate all programs of length up to $m$ and count their distinct outputs. Since it is c.e., $Q(m, k)$ is many-one reducible to the Halting Problem $H$.
    By our assumption $H \in \PP^{F_{\C_U}}$, there exists a polynomial-time procedure to decide $Q(m, k)$ using the oracle $F_{\C_U}$.
    
    We can now compute the exact value of $N_m$ using binary search on the range $[0, 2^{m+1}]$:
    \begin{enumerate}
        \item Initialize the search range with $L=0$ and $R=2^{m+1}$. Initialize the answer $ans = 0$.
        \item While $L \le R$, perform the following steps:
        \begin{itemize}
            \item Let $mid = \lfloor (L+R)/2 \rfloor$.
            \item Query the oracle to check if $Q(m, mid)$ is true (i.e., is $N_m \ge mid$?).
            \item If the answer is \textbf{yes}: The true count is at least $mid$. Update the tentative answer $ans \gets mid$ and search the upper half by setting $L = mid + 1$.
            \item If the answer is \textbf{no}: The true count is strictly less than $mid$. Search the lower half by setting $R = mid - 1$.
        \end{itemize}
        \item Return $ans$.
    \end{enumerate}
    This algorithm performs $O(m)$ queries. Since the input length is $m$, and each query takes polynomial time, the total running time is polynomial in $m$. Thus, a machine $M$ solving \texttt{COUNT-SIMPLE} exists.
\end{proof}

\subsection*{Proof of Proposition \ref{prop:wn_complexity}}

\begin{proof}[Proof of Proposition \ref{prop:wn_complexity}]
The upper bound is immediate; we prove the lower bound.

Let $p$ be a shortest program for $w_n$ given $n$. Let $|p| = n - c$ for some $c$.

Observe that the value $n$ can be reconstructed from the pair $( p, c )$ (specifically, $n = |p| + c$). Consider the following constructive procedure given $p$ and $c$:
\begin{enumerate}
    \item Compute $n = |p| + c$.
    \item Compute $w_n = U(p, n)$.
    \item Run the enumeration of the set $S$ until the pair $( w_n, n )$ appears.
    \item At this point, the list of all strings with complexity $\le n$ is complete. Output the lexicographically first string $z$ not present in this list.
\end{enumerate}
By definition, the resulting string $z$ must satisfy $\C_V(z) > n$.
On the other hand, $z$ is specified by the pair $( p, c )$. The complexity of this description is:
\[ \C_V(z) \le |p| + O(\log c) = (n - c) + O(\log c). \]
Combining these inequalities, we get:
\[ n < \C_V(z) \le n - c + O(\log c). \]
This implies $c \le O(\log c)$, which is only possible if $c = O(1)$. Thus, $|p| \ge n - O(1)$.
\end{proof}

\subsection*{Proof of Proposition \ref{prop:tail_bound_overview}}

\begin{proof}[Proof of Proposition \ref{prop:tail_bound_overview}]
Let $R = \operatorname{Rest}(n, u)$. Let $h$ be the maximal integer satisfying:
\[ R \ge 2^{u - n + h}. \]
Our goal is to prove that $h = O(1)$.

Let $N_u$ be the number of strings of complexity at most $u$. Note that
$ N_u = 2^{u + O(1)}.$

We claim that $w_n$ can be reconstructed given $n, u$, and the first $k = n - h + O(1)$ bits of the binary representation of $N_u$.
Indeed, by using this information we can enumerate all strings of complexity at most $u$ except the last portion of size at most $2^{u-n+h}$.

By definition of $h$, all strings of complexity at most $n$ will be enumerated within the known portion. We identify $w_n$ as the last string in this enumeration.

\paragraph{Complexity Analysis.}
Since $u$ is computable from $n$ in $O(1)$, the description length of $w_n$ given $n$ is dominated by the prefix of $N_u$:
\[ \C_V(w_n \cnd n) \le k + O(1) = n - h + O(1). \]
We now apply the lower bound from Proposition \ref{prop:wn_complexity}, which states $\C_V(w_n \cnd n) \ge n - O(1)$. Combining these inequalities we get that $h = O(1)$.
\end{proof}

\end{document}